\documentclass{elsarticle}
\usepackage[margin=1in]{geometry}

\usepackage{lineno}

\usepackage{amscd}
\usepackage{amsmath}
\usepackage{amsthm}
\usepackage{amssymb}
\usepackage{ifthen}
\usepackage{eucal} % for mathcal
\usepackage{url}
\usepackage{clrscode}
\usepackage{multirow}
\usepackage{graphicx}
\usepackage{proof}
\usepackage{times}
\usepackage{amsthm}
\usepackage{tikz}
\usepackage{tikz-cd}
\usepackage{epsfig}
\usepackage{subcaption}
\usepackage{shortlst}
\usepackage{multicol}
\usepackage{prelimdraft}
\usepackage{siunitx}
\newcommand{\R}{\ensuremath{\mathbb{R}}}

\newcommand{\Q}{\ensuremath{\mathbb{Q}}}
\newcommand{\Z}{\ensuremath{\mathbb{Z}}}

\newcommand{\cplusplus}{C\raisebox{0.5ex}{\small ++}}
\newcommand{\Filt}[1]{\ensuremath{\leq_{#1}}}

\newcommand{\complexity}[1]{\textsc{#1}}
\newcommand{\NPH}{\complexity{NP-Hard}}

\newcommand{\NPC}{\complexity{NP-Complete}}

\newcommand{\ablp}{\textsc{$\alpha$-Balanced-Minimum-Blowup}}

\newcommand{\avertex}{\textsc{$\alpha$-Subgraph-Balanced-Vertex-Separator}}

\newcommand{\bunny}{{\textsf{B}}}
\newcommand{\sphere}{\textsf{S}}
\newcommand{\clique}{\textsf{C}}
\newcommand{\multiblob}{\textsf{M}}
\newcommand{\blobs}{\multiblob}
\newcommand{\gnp}{\textsf{G}}

\newcommand{\Erdos}{Erd\H{o}s}
\newcommand{\Renyi}{R{\'e}nyi}

\newcommand{\bd}{\ensuremath{\partial}}
\DeclareMathOperator{\im}{im}

\newcommand{\betti}{\ensuremath{\beta}}
\newcommand{\tensor}{\otimes}

\newcommand{\Cl}[1]{\operatorname{Cl}{#1}}

\newcommand{\card}[1]{|#1|}
\newcommand{\mv}{Mayer-Vietoris }
\newcommand{\mvb}{\mv blowup complex}
\newcommand{\K}{K}

\newcommand{\C}{U}
\newcommand{\N}{N}
\newcommand{\M}{M}
\newcommand{\factor}{\card{\K^{\C}}/\card{\K}}
\newcommand{\ratio}{\factor}
\newcommand{\Parfor}{\kw{parallel for}}
\newtheorem{theorem}{Theorem}
\newtheorem{lemma}{Lemma}

\theoremstyle{definition}
\newtheorem{example}{Example}[subsection]

\definecolor{afra}{RGB}{198,168,208}
\definecolor{afrablue}{RGB}{143,166,215}
\definecolor{afragreen}{RGB}{182,215,112}
\definecolor{afrapurple}{RGB}{218,177,239}
\definecolor{darkgray}{gray}{0.3}
\definecolor{afrapurplelight}{RGB}{198,168,208}
\definecolor{afrabluedark}{RGB}{75,113,191}
\definecolor{afrabluelight}{RGB}{143,166,215}
\definecolor{afragreendark}{RGB}{154,191,75}
\definecolor{afragreenlight}{RGB}{182,215,112}
\definecolor{afrapurpledark}{RGB}{162,75,191}

\usepackage{xcolor}
\usepackage{pgfplots}
\usepackage{pgfplotstable}
\usepackage{filecontents}
\usetikzlibrary{calc, fit, shadows, arrows, positioning, patterns, external, matrix}
\pgfdeclarelayer{edges}
\pgfdeclarelayer{quadcell}
\pgfsetlayers{quadcell,edges,main}
\pgfplotsset{compat=newest}

\makeatletter \newcommand \listoftodos{\section*{Todo list} \@starttoc{tdo}}
\newcommand\l@todo[2]
  {\par\noindent \textit{#2}, \parbox{10cm}{#1}\par} \makeatother

\graphicspath{{figs/},{figs/speedup-figs/}}

\def\lbl#1#2{\begingroup
	    #2%
	    \def\@currentlabel{#2}%
	    \phantomsection\label{#1}\endgroup
	}

\usepackage{hyperref}
\title{Multicore Homology via \mv{}}
\begin{document}
%\date{\today}
\DRAFT
\begin{frontmatter}
 \author{Ryan H. Lewis\corref{cor}}
 \ead{rhl@stanford.edu}
 \address{Institute for Computational Mathematics and Engineering, Huang Building, Stanford University}
  \cortext[cor]{Principal Corresponding author}
  
\author{Afra~Zomorodian\corref{cor2}}
 \ead{afra@cs.dartmouth.edu}
 \address{The D. E. Shaw Group, 1166 Avenue of the Americas, Ninth Floor, New York, New York}
 
\begin{abstract}
In this work we investigate the parallel computation of homology using the \mv principle. We present a 
two stage approach for parallelizing persistence. In the first stage, we produce a cover of the input cell complex by overlapping subspaces.
In the second stage, we use this cover to build the \mvb{}, a topological space, which organizes the various subspaces needed for employing the \mv principle. 
Next, we compute the homology of each subspace in the blowup complex in parallel and then glue these results together in serial. 
We show how to use the persistence algorithm to organize these computations.
In the first stage, any algorithm can be used to produce a cover of the input complex. We describe an algorithm for producing a cover of a space with a simple structure and 
bounded overlap based on graph partitions. Additionally, we present a simplistic model for the problem of finding covers appropriate for parallel algorithms and show that finding such covers is \NPH{}.  
Finally, we present a second parallel homology algorithm. This algorithm avoids the explicit construction of the blowup complex saving space. 
We implement our algorithms for multicore computers, and compare them against each other as well as existing serial and parallel algorithms with a suite of experiments. 
We achieve roughly $8 \times$ speedup of the homology computations on a 10-dimensional complex with about 46 million simplices using 11 cores.
\end{abstract}
\end{frontmatter}
\section{Introduction}
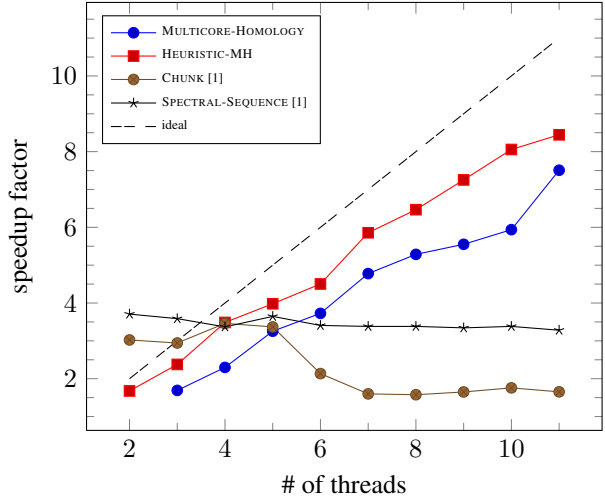
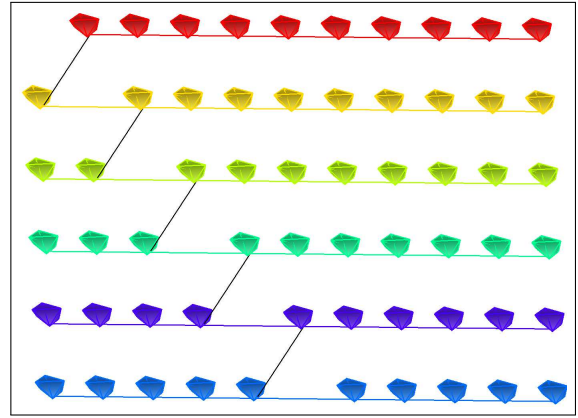
\begin{figure}
\pgfplotstableread[col sep=comma]{pgf-speedup-figs/results/cover_homology/clique.11.22720.csv}\tablecovermblb
\pgfplotstableread[col sep=comma]{pgf-speedup-figs/results/concurrent_homology/clique.11.22720.csv}\tableconcurrentmblb
\pgfplotstableread[col sep=comma]{pgf-speedup-figs/results/phat_14_chunk/clique.11.22720.csv}\tablephatchunkmblb
\pgfplotstableread[col sep=comma]{pgf-speedup-figs/results/phat_14_ss/clique.11.22720.csv}\tablephatssmblb
%%compute first regression 
\pgfplotstablecreatecol[linear regression={x=num_threads, y=speedup}]{regression}{\tablecovermblb}
\xdef\slopeA{\pgfplotstableregressiona} 
\xdef\interceptA{\pgfplotstableregressionb}
%%compute second regression 
\pgfplotstablecreatecol[linear regression={x=num_threads, y=speedup}]{regression}{\tableconcurrentmblb}
\xdef\slopeB{\pgfplotstableregressiona} 
\xdef\interceptB{\pgfplotstableregressionb}
%%make percentages
\pgfmathparse{\slopeA*100}
\pgfmathprintnumberto[precision=0]{\pgfmathresult}{\efficiencyA}
\pgfmathparse{\slopeB*100}
\pgfmathprintnumberto[precision=0]{\pgfmathresult}{\efficiencyB}
%\fbox{
\begin{subfigure}[t]{.45\linewidth}
\begin{tikzpicture}[scale=1]
\begin{axis}[legend cell align=left,xlabel=\# of threads, ylabel=speedup factor, minor y tick num={3}, minor x tick num={1},legend style={legend pos=north west, font=\tiny}]
\legend{$\proc{Multicore-Homology}$,  $\proc{Heuristic-MH}$, $\proc{Chunk}$~\cite{bkr-cccph-13}, $\proc{Spectral-Sequence}$~\cite{bkr-cccph-13}, ideal}
\addplot table [x=num_threads, y=speedup, col sep=comma, skip coords between index={0}{1}] from  {\tableconcurrentmblb};
\addplot table [x=num_partitions, y=speedup, col sep=comma, skip coords between index={0}{1}] from {\tablecovermblb};
\addlegendentry{Heuristic \mv{}} %
\addplot table [x=num_partitions, y=speedup, col sep=comma, skip coords between index={0}{1}] from {\tablephatchunkmblb}; 
\addplot table [x=num_partitions, y=speedup, col sep=comma, skip coords between index={0}{1}] from {\tablephatssmblb};
\addplot[dash pattern=on 4pt off 1pt on 4pt off 4pt, domain=2:11]{x};
\addlegendentry{ideal speedup}
\end{axis}
\end{tikzpicture}
\caption{We achieve ${\sim}\efficiencyA$\% efficiency across 11 cores via \proc{Heuristic-MH} and ${\sim}\efficiencyB\%$ via \proc{Multicore-Homology}. 
We also plot the speedup factor of the parallel algorithms \proc{Chunk} and \proc{Spectral-Sequence} from Bauer et. al~\cite{bkr-cccph-13}.}
\label{fig:blobs-speedup}
\end{subfigure}
%}
\hspace{1cm}
%\fbox{
\begin{subfigure}[t]{.45\linewidth}
\vspace{-1.27cm}
\begin{tikzpicture}[scale=1]
\fbox{\includegraphics[width=.98\linewidth]{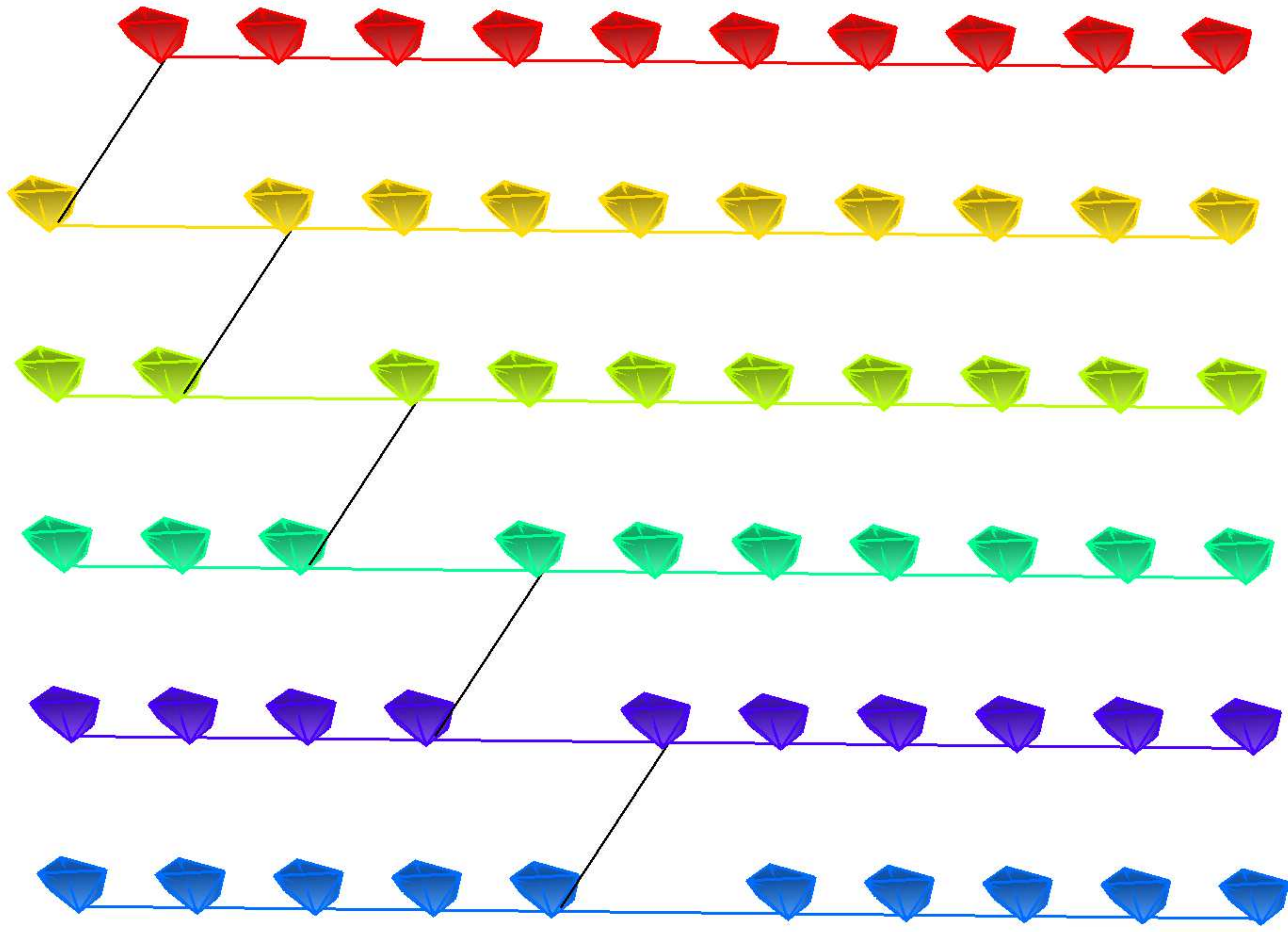}\vspace{1cm}}
\end{tikzpicture}
\vspace{1cm}
\caption{Shown is a portion of the simply connected input space $\multiblob$. {\blobs} contains ${\sim}22$K copies of a fully connected 10 dimensional complex on 11 vertices each connected to the next by a single edge.
Each color represents a portion of 7 sets of a cover by 12 pieces.}
\label{fig:blobs-vis}
\vspace*{\baselineskip}
\end{subfigure}
%}
\caption{On the left is the speedup in homology computation for the 10 dimensional complex with 45M simplices partially shown on the right.}
\label{fig:front-picture}
\end{figure}
In this paper, we present fast multicore algorithms for computing
the homology of arbitrary dimensional cell complexes over field
coefficients. Figure~(\ref{fig:front-picture}) shows the speedup factor
of our two algorithm for computing homology over $\Z_2$ coefficients
of the data set $\multiblob$, described in Section~\ref{sec:exp}. By decomposing the
space into the 11 pieces visualized in Figure~(\ref{fig:blobs-vis}),
we are able to reduce the boundary matrix of the input space in 
.37 seconds, approximately eight times faster than the 3 seconds 
necessary for serial computation.
All our timings are done on a 64-Bit GNU/Linux machine with dual, six core, 
2.93Ghz Intel X5670 CPUs, and hyperthreading disabled.

\subsection{Motivation}
We are motivated by \emph{topological data analysis} 
which attempts to extract a topological understanding of scientific data from 
finite sets of samples. Usually data analysis assumes that the input point cloud
comes from some underlying geometric space. Topological data analysis focuses on
the recovery of the lost topology of this underlying space \cite{c-tnd-09}. 
The classic pipeline for topological data analysis follows a two step process. First, we 
compute a combinatorial model approximating the structure of the underlying 
space. Second we compute topological invariants on these structures. One popular
invariant, persistent homology \cite{elz-tps-02, zc-cph-05}, captures multiscale
topological structure. Computing field homology, 
especially over $\Z_2$ coefficients, is an integral part of topological data 
analysis.

In this paper, we focus on developing a parallel algorithm to compute homology 
on multicore shared-memory machines. This algorithm is a first step toward a 
distributed-memory algorithm that will allow us to compute the persistent homology of 
massive structures on computer clusters.

\subsection{Prior Work}
There is a large literature on serial computation of integer homology.
Dumas et al.\@ review algorithms for computing integer homology that take 
advantage of the sparsity of boundary matrices derived from simplicial 
complexes~\cite{d-snf-03}. Their software is available within the GAP software 
package~\cite{GAP4}. Joswig surveys the computation of invariants, including 
homology for simplicial spaces with a focus on manifolds~\cite{j-csm-04}.
Kaczy{\'n}ski et al.\@ develop heuristics to compute cubical 
homology~\cite{kmm-ch-04}. Kaltofen et al.\@ provides a theoretical
investigation of randomized parallel algorithms for computing the Smith normal 
form~\cite{kks-snf-87,kks-snf-89} over finite fields and $\Q$, however, 
these algorithms are not useful in practice~\cite{ProofByAuthority}.

Any parallel computation of homology would require a decomposition of the
space into pieces.  The theory of \emph{spectral sequences} explains how to 
compute the homology of a space from its pieces. In this work, we decompose 
our input space using a \emph{cover} so the pieces correspond to \emph{subspaces} and their various intersections. 
The \mv spectral sequence expresses the relationship between the homology of these subspaces to the homology of the space itself.
This makes the \mv spectral sequence a natural gadget to study when developing algorithms for parallel homology~\cite{hatcher}. 
Merino et al.\@ use the \mv exact sequence to compute the homology of
three-dimensional simplicial complexes~\cite{bmlf-cmv-10}.  
Lipsky et al.\@ use the \mv spectral sequence in an attempt to derive a parallel algorithm~\cite{lsv-ss-11}.
Both works are theoretical in nature. The researchers do not address algorithmic issues of complexity, finding
covers for input, implementations of their algorithms, or any empirical results. 
The \mvb{} is the \emph{total complex} of the \mv spectral sequence. In other words the \mvb{} is a topological space which
encodes the data given as input to the spectral sequence. Its homology is equivalent to that of the original input space. 
Zomorodian and Carlsson show how computing homology of \mvb{} localizes the homology basis~\cite{zc-lh-08}.   

The \mv spectral sequence is not the only algebraic tool which is useful for parallel homology computation.
The spectral sequence of a filtration shows how a sequence of \emph{relative homology} computations may be carried out
in parallel on contiguous chunks of a boundary matrix to arrive at the homology of a space. 
Bauer, et al.\@ explore this approach to computing homology in parallel~\cite{bkr-cccph-13}. 

\subsection{Our Work}
In this paper we design and implement a divide and conquer framework for computing the field 
homology of a cellular space in parallel. Field homology is popular in 
topological data analysis since it can be computed in polynomial time and the persistence
algorithm exhibits linear-time behavior in practice~\cite{elz-tps-02,zc-cph-05}. 
Our framework relies on the \mvb{}, a spatial version of the \mv spectral sequence~\cite{zc-lh-08}. 
The \mvb{} is the \emph{total complex} of the terms of the first page of the 
\mv spectral sequence and its homology groups are isomorphic to that of the original space. 
In this work we show how to build the \mvb{} and compute its homology in parallel using the persistence algorithm. 
We note that while we restrict our attention to field homology our software could be modified to produce $\Z$-valued homology.

Our approach has two stages. In the first stage, we find a cover of the input space. In the second stage, we 
use this cover to build the blowup complex and compute its homology in parallel. The homology computation within the second stage 
may be viewed as the parallel computation of relative homology on chunks of the boundary matrix for the blowup complex. However, 
because of the structure of the blowup complex many relative computations are the same as their non-relative siblings that is,
they do not need to be further reduced against each other.

Since the first stage of the pipeline requires a cover of the input space, we investigate the general problem of finding covers of spaces.
In Section~\ref{sec:covers}, we identify a class of covers which lend themselves to efficient parallel
algorithms and model the problem of finding covers in this class as an optimization problem. 
We then show that solving this problem is \NPH{}. Motivated by this result, we instead provide a 
algorithm for producing covers with bounded overlap based on graph partitioning in Section~\ref{sec:pcover}.  
We may avoid building the blowup complex by using the cover to generate a new filtration on the original space for 
carrying out parallel computations without the blowup complex.

In Section~\ref{sec:exp}, we present the results of a suite of experiments using a multicore version of our parallel algorithms and 
provide experimental results. All of the techniques in this paper are deterministic. Our software and our datasets are 
\href{http://ctl.appliedtopology.org}{publicly available}.

\section{Background}
We begin with a review of simplicial complexes, homology, and blowup complexes.
We refer the read to Hatcher for background material in algebraic topology~\cite{hatcher}.
and to Zomorodian~\cite[Chapter 13]{z-ct-10} for computational topology. 
In principle the methods outlined in this paper generalize to any type of cellular space, 
however we restrict ourselves to simplicial complexes. 

\subsection{Preliminaries}
Let $[n] = \{ 0,1, \ldots ,n \}$ be the first $n+1$ natural numbers. This
definition is not conventional but we adopt the notation used
in previous work for continuity with prior work~\cite{zc-lh-08}. A \emph{multiset}
is a pair $(A,c)$ where $c: A \rightarrow \mathbb{N}$. A \emph{decomposition}
of a set $S$ is a collection of nonempty subsets of $S$ whose union is $S$.
A \emph{partition} of a set $S$ is a decomposition of $S$ by disjoint sets.
A \emph{graph} $G = (V,E)$ is a set $V$ of \emph{vertices}, and 
a set $E \subseteq V \times V$  of \emph{edges}. 
Suppose we have a graph $G = (V,E)$. A \emph{graph partition} is a partition 
$P = \{P_i\}_{i \in [n-1]}$ of $V$ into $n$ subsets. 
A \emph{cut} is a partition of $V$ into two sets $A$ and $B$. A \emph{vertex separator}
of a graph $G$ is a set of vertices $I$ such that the removal of $I$ from $G$ results
in a disconnected graph.  

A \emph{simplicial complex} is a collection $\K$ of finite sets called
\emph{simplices} such that if $\sigma \in \K$ and $\tau \subseteq \sigma$ then 
$\sigma \in \K$. 
We say that $\tau$ is a \emph{face} of $\sigma$, its \emph{coface}. A simplex
is \emph{maximal} if it has no proper coface in $\K$. The set of maximal cells of a 
simplicial complex $\K$ is $\M(\K)$. If $\card{\sigma} = k+1$
then $\sigma$ is a $k$-simplex, it has \emph{dimension} $k$, denoted 
$\dim{\sigma} = k$. We say that $\K$ is \emph{$d$-dimensional} if 
$d = \max_{\sigma \in \K} \dim{\sigma}$. 
Given a simplicial complex $\K$ the set of maximal 
cells can be enumerated in $O(md)$ time.

Suppose we have a subset $L \subseteq K$.  $L$ is a \emph{subcomplex}
if it is a simplicial complex. The \emph{closure} of $L$ is 
$\Cl(L) = \{ \tau \mid \tau \subseteq \sigma \in L\}$ and is a simplicial complex. The 
\emph{$k$-skeleton} of a complex $\K$ is the set of all simplices
of dimension less than or equal to $k$. Note that the 1-skeleton of
any complex may be viewed as a graph. 
Let $\Delta^n$ be the $n$-simplex defined on $[n]$.  We note that 
$\Delta^n$ is traditionally defined in a geometric setting and is 
called the \emph{standard $n$-simplex}~\cite{hatcher}, although we 
are using an abstract version here for our purposes. For any
\emph{indexing set} $J \subseteq [n]$, $\Delta^J$ is the $(\card{J}-1)$ 
dimensional face of $\Delta^n$  that is defined on $J$.
We define a \emph{filtration} of $\K$ to be a partial ordering on the simplices of
$\K$ such that every prefix of the ordering is a subcomplex and denote
it as $\Filt{\K}$.
Given a simplicial complex $\K$, An \emph{open cover}
of $\K$ is a decomposition of $\K$ and when each
cover set is closed we call the cover a \emph{closed cover} $\C$. 
Except where explicitly specified all covers in this work are closed. The \emph{nerve} $\N(\C)$
of a cover $\C$ is the simplicial complex on $[\card{\C}-1]$ whose $k$-simplices 
represent the non-trivial intersections of subsets of $\C$ of size $k+1$.
The nerve is a subcomplex of the standard $n$-simplex and so we 
denote its simplices by $\Delta^J$ where $J \subseteq [\card{\C}-1]$. 
It is convenient to encode the cover $\C$ as a map from $\K$ to $\N(\C)$ where each 
simplex $\sigma \in \K$ is mapped to $\N(\sigma)$ the simplex in $\N(\C)$ which lists
the cover sets containing $\sigma$.

A simplicial complex may be viewed as the result of gluing simplices of 
different dimensions along common faces. Other types of complexes are defined 
similarly using different types of \emph{cells}. 
Such \emph{cellular} complexes include $\Delta$-complexes, \emph{cubical} 
complexes, \emph{simplicial sets}, and \emph{CW-complexes}, 
to name a few~\cite{ez-ssc-50,hatcher,kmm-ch-04,m-soat-68}.
In this paper, we restrict to simplicial complexes as input, although our 
methods generalize easily to other types of complexes. 

\subsection{Homology}
In this section, we describe the homology of cellular spaces over 
field coefficients. Homology, however, is an invariant of arbitrary topological 
spaces and may be computed over arbitrary coefficient rings~\cite{hatcher}.  
Suppose we are given a finite cellular complex $\K$ and a field $k$. 
The \emph{$n$th chain vector space $C_n$} is the $k$-vector space generated by 
the set of $n$-dimensional cells of $K$, its \emph{canonical basis}.  
Suppose we are given a linear \emph{boundary operator} 
$\bd_n\colon C_n \rightarrow C_{n-1}$ such that 
$\bd_n \circ \bd_{n-1} \equiv 0$ for any $n$.  
The boundary operator connects the chain vector space into a 
\emph{chain complex $C_*$}:
\begin{linenomath*}
\begin{equation*}
  \cdots \rightarrow             C_{n+1}
         \xrightarrow{\bd_{n+1}}  C_n
         \xrightarrow{\bd_n}     C_{n-1}
         \rightarrow \cdots .
\label{eqn:chaincomplex}
\end{equation*}
\end{linenomath*}
Given any chain complex, the \emph{$n$th homology vector space $H_n$} is:
\begin{linenomath*}
\begin{equation}
  \label{eqn:homology}
  H_n = {\ker{\bd_n}}\,/\,{\im{\bd_{n+1}}}, 
\end{equation}
\end{linenomath*}
where $\ker(.)$ and $\im(.)$ are the \emph{kernel} and \emph{image} of $\bd$, 
respectively.
Each homology vector space is characterized fully by its \emph{Betti number}, 
$\betti_n = \dim{H_n}$.  
We now only need to define boundary operators to get homology. 
For simplicial homology, we begin by defining the action of the boundary operator on any
$n$-simplex $[v_0,\ldots,v_n] \in \K$:
\begin{linenomath*}
\begin{equation*}
\bd_n [v_0,\ldots,v_n] = \sum_i (-1)^i [v_0,\ldots,\hat{v_i},\ldots,v_n],
\end{equation*}
\end{linenomath*}
where $\hat{v_i}$ indicates that $v_i$ is deleted from the vertex 
sequence. The boundary operator is the linear extension of the above action.

Over field coefficients, homology is a vector space characterized by its dimension, 
so we may compute homology using \emph{Gaussian elimination}~\cite{uhlig}.  
In practice, we use the 
\emph{persistence algorithm}~\cite{elz-tps-02,zc-cph-05}.
This algorithm can compute the homology of any \emph{based persistence complex}
~\cite{zc-lh-08},
a class that includes simplicial complexes as well as the 
blowup complex. As input, this algorithm requires a basis for 
the chain complex $C_*$, a boundary operator $\bd_n$, and a filtration 
on the basis elements. The algorithm proceeds by determining if
the addition of a cell into the complex creates a new homology class
or annihilates a homology class previously created. The result is 
a pairing between cells which create homology and the corresponding
cell which destroy's that homology. Except, if a homology class is never killed, in which case it is left unpaired.
$\betti_i$ is the number of unpaired $i$-cells.

We focus on characterizing the three inputs needed for computing the homology 
of a blowup complex using the persistence algorithm.
\subsection{Blowup Complex}
\begin{figure*}
\begin{subfigure}[t]{.33\textwidth}
\centering
%\hspace*{-1cm}%
\begin{tikzpicture}[scale=.5, y=0.6pt, x=.75pt]
     %row 1
     \node[font=\large] at (-50, 800) {$(K, U)$};
     \draw[fill=afragreen, draw = black,  line width=2]  (47,800) circle (10pt);  
     \draw[fill=afragreen, draw = black, line width=2]  (117,800) circle (10pt);   
     \draw[fill=afragreen, draw = black, line width=2]  (190,800) circle (10pt); 
     \draw[fill=afragreen, draw = black, line width=2]  (261,800) circle (10pt); 
     \begin{pgfonlayer}{edges}
            \path[draw=black,fill=black,line width=2] (47,800) -- (117, 800);
            \path[draw=black,fill=black,line width=2] (117,800) -- (190, 800);
            \path[draw=black,fill=black,line width=2] (190,800) -- (261, 800);
      \end{pgfonlayer}      
      \draw[draw, color=afrablue, fill=none, line join=round,draw opacity=0.978,line width=2] (117, 800) ellipse (100 and 45);
      \draw[draw, color=afrapurple, fill=none, line join=round,draw opacity=0.978,line width=2] (190, 800) ellipse (100 and 45);
         %row 2 
    \node[font=\large] at (-40, 740) {$K^0$};
     \draw[fill=afragreen, draw = black, line width=2]  (47,730) circle (10pt);  
     \draw[fill=afragreen, draw = black, line width=2]  (117,730) circle (10pt);   
     \draw[fill=afragreen, draw = black, line width=2]  (190,730) circle (10pt); 
     \begin{pgfonlayer}{edges}
            \path[draw=black,fill=black,line width=2] (47,730) -- (117, 730);
            \path[draw=black,fill=black,line width=2] (117,730) -- (190, 730);
      \end{pgfonlayer}   
               %row 3
     \node[font=\large] at (-50, 680) {$K^1$};
     \draw[fill=afragreen, draw = black, line width=2]  (117,680) circle (10pt);   
     \draw[fill=afragreen, draw = black, line width=2]  (190,680) circle (10pt); 
      \draw[fill=afragreen, draw = black, line width=2]  (261,680) circle (10pt); 
         \begin{pgfonlayer}{edges}
            \path[draw=black,fill=black,line width=2] (117,680) -- (190, 680);
            \path[draw=black,fill=black,line width=2] (190,680) -- (261, 680);
      \end{pgfonlayer}   
          %row 4
\node[font=\large] at (-50, 620) {$K^{[1]}$};
     \draw[fill=afragreen, draw = black, line width=2]  (117,620) circle (10pt);   
     \draw[fill=afragreen, draw = black, line width=2]  (190,620) circle (10pt); 
         \begin{pgfonlayer}{edges}
            \path[draw=black,fill=black,line width=2] (117,620) -- (190, 620);
      \end{pgfonlayer}     
\end{tikzpicture}
\caption{Space and Cover}
\label{fig:space-n-cover}
\end{subfigure}
%\hfill
\begin{subfigure}[t]{.33\textwidth}
\centering
\begin{tikzpicture}[scale=.5, y=0.6pt, x=.75pt]
    %disj 0
    \begin{scope}[shift={(143,80)}]
    \node[font=\large] at (300, 730) {$K^0 \times \Delta^{0}$};
     \draw[draw, color=afrablue, fill=none, line join=round,draw opacity=0.978,line width=2] (117, 730) ellipse (100 and 45);
     \draw[fill=afrablue, draw = black, line width=2]  (47,730) circle (10pt);  
     \draw[fill=afrablue, draw = black, line width=2]  (117,730) circle (10pt);   
     \draw[fill=afrablue, draw = black, line width=2]  (190,730) circle (10pt); 
     \begin{pgfonlayer}{edges}
            \path[draw=black,fill=black,line width=2] (47,730) -- (117, 730);
            \path[draw=black,fill=black,line width=2] (117,730) -- (190, 730);
      \end{pgfonlayer}   
      \end{scope}
      %disj  1
       \begin{scope}[shift={(0,0)}]
         \node[font=\large] at (380, 680) {$K^1 \times \Delta^{1}$};
       \draw[draw, color=afrapurple, fill=none, line join=round,draw opacity=0.978,line width=2] (190, 680) ellipse (100 and 45);
     \draw[fill=afrapurple, draw = black, line width=2]  (117,680) circle (10pt);   
     \draw[fill=afrapurple, draw = black, line width=2]  (190,680) circle (10pt); 
      \draw[fill=afrapurple, draw = black, line width=2]  (261,680) circle (10pt); 
         \begin{pgfonlayer}{edges}
            \path[draw=black,fill=black,line width=2] (117,680) -- (190, 680);
            \path[draw=black,fill=black,line width=2] (190,680) -- (261, 680);
      \end{pgfonlayer}   
            \end{scope}
\end{tikzpicture}
\caption{Local pieces of the blowup complex.}
	\label{fig:local-pieces}
\end{subfigure}
%\hfill
\begin{subfigure}[t]{.33\textwidth}
\centering
\begin{tikzpicture}[scale=.5, y=0.6pt, x=.75pt]
    %disj 0
    \begin{scope}[shift={(143,80)}]
    \node[font=\large] at (300, 730) {$K^0 \times \Delta^{0}$};
     \draw[draw, color=afrablue, fill=none, line join=round,draw opacity=0.978,line width=2] (117, 730) ellipse (100 and 45);
     \draw[fill=afrablue, draw = black, line width=2]  (47,730) circle (10pt);  
     \draw[fill=afrablue, draw = black, line width=2]  (117,730) circle (10pt);   
     \draw[fill=afrablue, draw = black, line width=2]  (190,730) circle (10pt); 
     \begin{pgfonlayer}{edges}
            \path[draw=black,fill=black,line width=2] (47,730) -- (117, 730);
            \path[draw=black,fill=black,line width=2] (117,730) -- (190, 730);
      \end{pgfonlayer}   
      \end{scope}
      %disj  1
       \begin{scope}[shift={(0,0)}]
         \node[font=\large] at (380, 680) {$K^1 \times \Delta^{1}$};
       \draw[draw, color=afrapurple, fill=none, line join=round,draw opacity=0.978,line width=2] (190, 680) ellipse (100 and 45);
     \draw[fill=afrapurple, draw = black, line width=2]  (117,680) circle (10pt);   
     \draw[fill=afrapurple, draw = black, line width=2]  (190,680) circle (10pt); 
      \draw[fill=afrapurple, draw = black, line width=2]  (261,680) circle (10pt); 
         \begin{pgfonlayer}{edges}
            \path[draw=black,fill=black,line width=2] (117,680) -- (190, 680);
            \path[draw=black,fill=black,line width=2] (190,680) -- (261, 680);
      \end{pgfonlayer}   
            \end{scope}
      %blowup edges 
      \begin{scope}[shift={(143,80)}]
     \begin{pgfonlayer}{edges}
            \path[draw=black,fill=black,line width=2] (47,730) -- (47, 600);
            \path[draw=black,fill=black,line width=2] (117,730) -- (117, 600);
      \end{pgfonlayer}     
          \node[font=\large] at (280, 665) {$K^{[1]} \times \Delta^{[1]}$};
        \begin{pgfonlayer}{quadcell}
      \draw [fill=afragreen,  preaction={fill, afragreen}, pattern=north west lines, pattern color=black] (47, 600) rectangle (117, 730);
	\end{pgfonlayer}
	\end{scope}
\end{tikzpicture}
		\caption{The blowup complex.}
		\label{fig:blowup}
\end{subfigure}
%\begin{subfigure}
%\caption{Persistence Barcode. Colors represent homology groups computed
%	   on each piece of the blowup complex. Stacked colors represent parallelism}
%%	\label{fig:barcode}
%%	\def\svgwidth{1.1in}
%%	\input{figs/barcode.pdf_tex}
%%\end{subfigure}
\caption{Our approach. We are given a space equipped with a 
	 cover~(\protect\subref{fig:space-n-cover}), the former represented by a path with
         four vertices and three edges and the latter represented by ovals.  
	 First, at time $(t = 0)$ we blowup up the space into 
	 local pieces~(\protect\subref{fig:local-pieces}), each local piece is a copy of 
	 the corresponding cover set, then, at $(t = 1)$ we glue together 
	 duplicated simplices by adding in the blowup cells, rendering them 
	 homologically equivalent, which gives us the blowup 
	 complex~(\protect\subref{fig:blowup}).
}
\label{fig:vignette}
\end{figure*}
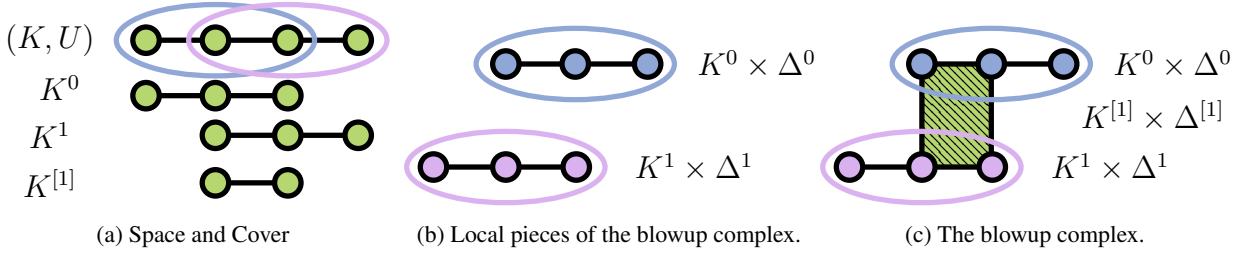
Like homology, the blowup complex may be defined for arbitrary topological 
spaces~\cite{zc-lh-08}, but in this paper we focus on blowups of simplicial 
complexes. For a longer exposition of the \mvb{} we refer the reader to Zomorodian
\& Carlsson~\cite{zc-lh-08}. Given a simplicial complex $K$ and cover 
$\C = \{\C_i\}_i$ of $n$ subcomplexes, let 
$\K^J = \bigcap_{k \in J} \C_j$. The \emph{Mayer-Vietoris blowup complex} is:
\begin{linenomath*}
\begin{align*}
\K^U &= \bigcup_{\emptyset \not = J \subseteq [n-1]} \K^J \times \Delta^J,
\end{align*}
\end{linenomath*}
where $\times$ is the Cartesian product~\cite{zc-lh-08} and $\Delta^J$ is a face of $\N(\C)$.
 % Do example here before any further talk.
\begin{example}
\label{ex:blowup}
Suppose we have a space $\K$ with cover $\C = \{ \C_0, \C_1 \}$ as is shown on 
the top of Figure~(\ref{fig:space-n-cover}), where we use a line as a 
representative space and ovals to indicate cover sets, and the four vertices
of the line are labeled from left to right as $a, b, c, d$ respectively.
The cover defines the intersection $\K^{[1]} = \K^{ \{ 0 , 1 \} }$. 
The corresponding blowup is shown in in Figure~(\ref{fig:blowup}). We list each of
the relevant pieces of $\K^{\C}$ as well as the nerve of the cover where we denote simplices as strings for 
brevity.  
\begin{linenomath*}
\begin{align*}
\N(\C) &= \{ 0, 1, 01 \} \\
\K^0 \times \Delta^{\{0\}} &= \{a,b,c,ab,bc\} \times \{0\}, \\
\K^1 \times \Delta^{\{1\}} &= \{b,c,d,bc,bd\} \times \{0\}, \\
\K^{[1]} \times \Delta^{[1]} &=\{b,c,bc\} \times \{01\}. 
\end{align*}
\end{linenomath*}
\end{example}
% key property
Our work is based on the following key property.
The blowup complex $\K^U$ has the same homology as its base complex $\K$ in any 
dimension: $H_n(\K^U) \cong H_n(\K)$ for any $n$~\cite[Lemma 1]{zc-lh-08}. 
Our approach then is to compute homology of the blowup complex 
instead of the base complex. The blowup has a structure that allows 
for computation in parallel, unlike the base complex.

% simplicial
To compute the homology of the blowup complex, we may interpret the definition 
above in two different ways.  
At the space level, we may view each cell of the blowup complex as a product 
of two simplices $\sigma \times \tau$, where 
$\sigma \in \K$ and $\tau \in \N(\C) \subseteq \Delta^{[n]}$.  
For example, the product of two edges, $bc \times 01$, gives us a 
quadrilateral cell in Example~\ref{ex:blowup}.  
While we may then triangulate the blowup complex to get a simplicial complex in order 
to compute its homology, this is computationally prohibitive, due to the need for triangulation. 
Luckily this approach is also not necessary. Alternatively, we 
examine the chain complex attached to the blowup complex.

% chain
A basis for $C_n(\K^\C)$ is the set composed of elements 
$\sigma \tensor \Delta^J$ for all $\emptyset \not = J \subseteq [n-1]$ and 
simplices $\sigma \in K^J$ where $\dim{\sigma} + \dim{\Delta^J} = n$.  The notation
$\tensor$ denotes \emph{tensor product}. Recall that the tensor product of two 
vector spaces is obtained by taking a quotient of the free vector space
on the cartesian product~\cite[\textrm{Page }218]{hatcher}.
We define the boundary operator as~\cite[\textrm{Lemma }4]{zc-lh-08}:
\begin{linenomath*}
\begin{align*}
\bd{\left (\sigma \tensor \Delta^J \right)}
&=
\bd{\sigma} \tensor \Delta^J + 
(-1)^{\dim{\sigma}}\sigma \tensor \bd\Delta^J.
\end{align*}
\end{linenomath*}
Here, we are defining a boundary operator for the blowup complex on the left 
using the boundary operators on the right, all of which are simplicial and were 
defined in the previous section. 
\begin{example}
The boundary of the quadrilateral cell 
$bc \tensor 01$: in Example~\ref{ex:blowup} is:
\begin{linenomath*}
\begin{align*}
\bd{\left (bc \tensor 01 \right)}
&=
\bd{(bc)} \tensor 01 - bc \tensor \bd(01) \\
&= c \tensor 01 - b \tensor 01 - bc \tensor 1 + bc \tensor 0.
\end{align*}
\end{linenomath*}
\end{example}
Having specified the basis for the chain complex and a boundary operator of 
the blowup complex, we now need a filtration on the basis elements in order to use the 
persistence algorithm. In principle an arbitrary filtration will do.
But for computing homology in parallel, we will specify a particular filtration 
whose structure mirrors the structure of the blowup complex. 
%We end this section by relating the \mvb{} to the \mv spectral sequence. For those uninterested
%in this connection they may continue to Section~\ref{sec:blowup_structure}.
%
%\begin{figure}[h!]
%\input{spectral_sequence}
%\caption{(Left) The general form for the $E^0$ page of the \mv{}Spectral Sequence. The superscript $0$ is 
%omitted from each term in the $E^0$ page. The direct sum of the vector spaces on each diagonal
%correspond to the chain groups of each dimension within the blowup complex.
%(Right) A pictorial form of the $E^0$ page of the \mv{}Spectral Sequence instantiated on Example~(\ref{ex:blowup}) }
%\end{figure}
%\subsubsection{The \mvb{} and the \mv{}Spectral Sequence}
\section{Blowup Structure}
\label{sec:blowup_structure}
%intro
The filtration of the blowup complex has two phases,  the \emph{local} and the
 \emph{global} phase. In the local phase, the complex explodes into multiple 
pieces, representing the disjoint union of each set in the cover, as in Figure~(\ref{fig:local-pieces}). 
This means that we have potentially multiple versions of a simplex 
if it lies in an intersection of two sets in the cover. For example, since edge $bc$ falls within both sets in the 
cover in Figure~(\ref{fig:space-n-cover}), it is represented by two cells 
$bc \times 0$ and $bc \times 1$. The pieces at the local stage are disjoint, so 
we may compute the homology of the pieces in parallel. 

The global phase specifies cells that glue the different versions of the 
original simplices together, rendering them homologically equivalent.  
For example, in Figure~(\ref{fig:blowup}), the cell 
$b \times 01$ connects $b \times 0$ and $b \times 1$.  

To describe this filtration on the blowup complex, we assume that we have an 
arbitrary filtration $\Filt{\K}$ on the simplices of our input complex $\K$.  In practice, 
we often label the vertices of a complex using numbers or letters and use the lexicographic 
ordering of the vertices to generate a filtration on the complex. We use the same procedure 
with $\N(\C)$ as its vertices are numbered by definition.

Given a filtration $\Filt{\K}$ on $\K$ and $\Filt{\N(\C)}$ on $\Delta^n$ we define a partial order $\Filt{{\K^\C}}$ by 
ordering all cells in the local phase before those in the global phase. This amounts to 
comparing two cells $\sigma \times \Delta^M$ and $\tau \times \Delta^N$ by  
comparing the second factor according to $\Filt{\N(\C)}$. We may complete this partial order to a filtration by then comparing the 
first factor according to $\Filt{\K}$. 
\begin{example}
Figure~(\ref{fig:blowup}) has the following filtration: 
\begin{linenomath*}
\begin{equation*}
(\overbrace{a \times 0, b \times 0 ,c \times 0,  ab \times 0, bc \times 0}^
{\textrm{Local Piece \#0 } (t=0)},
\overbrace{b \times 1 , c \times 1, d \times 1, bc \times 1, cd \times 1}^
{\textrm{Local Piece \#1 }(t=0)},
\overbrace{b \times 01, c \times 01, bc \times 01}^
{\textrm{Global Piece } (t=1)}).
\end{equation*}
\end{linenomath*}
\end{example}
%%algorithms figure
\begin{figure}
\centering
\begin{minipage}[t]{.45\textwidth}
\begin{codebox}
\Procname{$\proc{Multicore-Homology}(\K,p)$}
 \li  $\C \gets \proc{Cover}(\K, p)$
 \li  $\K^{\C} \gets \proc{Build-Blowup-Complex}(\K, \C)$
 \li  $\Parfor$   $\Delta^J \in \N(\C)$
 \li  \Do $\proc{Pair-Cells}(\K^{J} \times \Delta^{J})\footnotemark$
      \End
\li $\For$ $d > 0$
 \li \Do $\For$ $\Delta^J \in \N(\C)$ a $d$-cell.
 \li  \Do $\proc{Pair-Cells}(\Cl{(\K^{J} \times \Delta^{J})})$ 
\end{codebox}
\end{minipage}
\begin{minipage}[t]{.45\textwidth}
\begin{codebox}
\Procname{$\proc{Build-Blowup-Complex}(\K, \C)$}
\li $\K^{\C} \gets \emptyset$
\li $\Parfor$ $\sigma \in \K$
\li \Do $\For$ $\tau \subseteq \C[\sigma]$
\li \Do $\K^{\C} \gets \K^{\C} \cup (\sigma \times \tau)$ 
\end{codebox}
\end{minipage}
\caption{Psuedocode for computing the blowup complex and its homology in parallel. $\proc{Cover}$ can be any algorithm for generating a cover of $\K$ by $p$ subspaces. The procedure $\proc{Pair-Cells}$ is defined in the Computational Topology section of the Algorithms and Theory of Computation Handbook \cite[\textrm{Page } 3-17]{z-ct-10}.}
\label{fig:multicore-code}
\end{figure}
\footnotetext{When the list of cells given as input to \proc{Pair-Cells} is not a sub complex computation should be interpreted as relative homology computation by ignoring
elements of the boundary which are not given in the input.}
The Algorithm in Figure~(\ref{fig:multicore-code}) shows how to build the blowup complex and compute its homology in parallel. 
The procedure $\proc{Build-Blowup-Complex}$ runs in parallel and has parallel running time $O(2m/p + p)$ time where $m = \card{\K^{\C}}$ and $p$ is the number of processors available. In practice $\proc{Build-Blowup-Complex}$ not only produces a blowup complex but also the filtration of the blowup complex prescribed above. 

The size of the blowup complex depends on the cover. In the worst case, all of 
the simplices in a space $\K$ are contained within all $n$ sets of the cover 
$\C$. In this case, for each simplex $\sigma \in \K$ we have a corresponding 
product cell $\sigma \times \Delta^n$, which has $2^n$ faces. That is, the 
blowup complex \emph{blows up} $\K$ to be $2^n$ times larger, thus deserving 
its name. Therefore, it is imperative to find a cover which minimizes blowup.

\section{Covers}
\label{sec:covers}
Given a simplicial complex $\K$, our goal is to compute its homology.  
Our approach, as illustrated in Figure~(\ref{fig:vignette}), is to find a cover, 
build the associated blowup complex, and compute the homology of the blowup complex in 
parallel. We have now explained all the steps of this approach except how to find a cover. 
We begin in Section~\ref{sec:hardness} by identifying 
properties of covers that lead to efficient computation.
We state an optimization problem over covers which 
minimizes the size of the blowup of a complex. We then show that this optimization problem is \NPH{}.  
In Section~\ref{sec:partition-based-covers}, we describe an algorithm 
that generates covers which have a simple structure, and bounded overlap 
based on graph partitions. We end the section by showing how a partition of the 
0-cells of a complex can be lifted to a partition of a filtration on the complex 
which can be used to compute homology in parallel without building the blowup complex.
\subsection{Minimum Blowups}
\label{sec:hardness}
In this section, we formalize the problem of finding covers
that minimize blowup size. We show that this problem is \NPH{}, and its 
decision-variant, \NPC{}.

It should be clear that seek a cover which does not yield a large blowup complex. 
To quantify blowup, we define the $\emph{blowup factor}$ as the ratio: 
$\factor.$ We search for a cover $\C$ of size $p$ that minimizes the blowup 
factor. Since we intend to compute the homology of each cover set in parallel,
the number of cover sets should be the number $p$ of available 
processors. Finally, each cover set should be approximately the same 
size. There are many ways of modeling this last constraint. 
We model it by enforcing that no cover set should be larger than a fixed fraction
$\alpha$ of the size of the input complex, where $\alpha \in (\frac{1}{p},1)$. 
Putting together all of the desired properties of blowups, we have the following 
optimization problem stated for $p = 2$ and $\alpha \in (\frac{1}{2}, 1)$:
\begin{description}
\addtolength{\itemsep}{-.7\baselineskip}
\item[\textsc{Problem:}]  \ablp 
\item[\textsc{Instance:}] A simplicial complex $K$
\item[\textsc{Goal:}] Find a cover $\C$ of $\K$ with $2$ elements such that: 
\[ \max{\card{\C_i}} \leq \alpha\card{\K} \textrm{ and } \factor \textrm{ is minimized.} \]
\end{description}
Our goal is to show that this problem is $\NPH{}$ and its decision problem variant $\NPC{}$.
For the decision problem variant to be $\NPC{}$ we need to show that $\factor$ may
be evaluated in polynomial time. Recall that $\K^{\C}$ might be exponentially larger than $\K$. 
For covers by two sets we may employ the following lemma.
\begin{lemma}
\label{lem:char-blowup-sol}
Let $\K$ be a complex and let $\C$ be a cover $\K$ 
of size $p > 1$. Suppose that the intersection
of any three sets in $\C$ vanishes. Then
\begin{equation*}
\factor = 1 + 2\frac{\card{I}}{\card{\K}}. 
\end{equation*}
where $I = \bigcup_{i \neq j}{\C_i \cap \C_j}$.
\end{lemma}
\begin{proof}
This follows directly from the product cell definition of $\K^\C$.
\end{proof}
\noindent Now we observe an important necessary condition of optimal solutions to \ablp{}.
\begin{lemma}
\label{lem:blowup-sol-max-simplices}
Given a complex $\K$ and $\C = \{\C_1, \C_2\}$ be an optimal solution of  \ablp{}, then 
$\C$ is a partition of $\M(\K)$ the maximal cells of $\K$.
\end{lemma}
\begin{proof}
If $\sigma \in \C_i \cap \C_j$ is a maximal cell, then consider the cover
$\C'$ obtained by removing $\sigma$ from the set of larger
cardinality. $\C'$ is certainly a cover satisfying $\alpha$-balance but
by Lemma~\ref{lem:char-blowup-sol} the blowup factor has decreased
which contradicts the optimality of $\C$.
\end{proof}
\noindent Suppose the input to \ablp{} is a graph $G$. In this context any cover $\C$ of $G$ is a pair of subgraphs $G_1,G_2$.
Lemma~\ref{lem:blowup-sol-max-simplices} tells us that in any optimal solution the intersection $I = G_1 \cap G_2$ of these two 
subgraphs is a set of vertices. The requirement that $\C$ is a cover implies that $I$ is a vertex separator. In other words
given a vertex separator of a graph $G$ we may view it as a cover of that graph and vice versa. 
The equivalent problem for vertex separators is for any $\alpha \in (\frac{1}{2},1)$:
\begin{description}
\addtolength{\itemsep}{-.8\baselineskip}
\item[\textsc{Problem:}]  \avertex{}
\item[\textsc{Instance:}] A graph $G$
\item[\textsc{Goal:}] Find a vertex separator $(V_1,V_2,I)$ of $G$ such that: 
\[ \card{I} \textrm{ is minimized} \textrm{ subject to } \max_i{(\card{V_i} + \card{E_i})} + \card{I} \leq \alpha(\card{V}+\card{E})  \]
\end{description}
where $E_i$ is the set of edges with at least one endpoint in $V_i$. \avertex{} is \NPH{} for any $\alpha \in (\frac{1}{2},1)$ and its decision problem variant is \NPC{}~\cite{rhl-yaggpis-14}.
\begin{theorem}
For any $\alpha \in (1/2,1)$ the optimization problem \ablp{} is \NPH{} and its decision problem variant \NPC{}.
\end{theorem}
\begin{proof}
By restricting \ablp{} and \avertex{} are equivalent when the former is restricted to graph instances.
\end{proof}
This procedure shows us that finding covers of graphs with bounded overlap also identifies partitions of that graph.
In the next section we show how given a complex $\K$ and a partition of its 1-skeleton one can produce a cover of the entire complex with bounded overlap.
\subsection{Partition-Based Covers}
\input{algorithm}
\label{sec:partition-based-covers}
\label{sec:pcover}
\begin{figure}[h!]
\centering
\begin{minipage}[b]{.45\textwidth}
\begin{description}
\addtolength{\itemsep}{-.65\baselineskip}
\item[\small\textbf{Input:}] \small A complex $\K$, and a graph partition $P$.
\item[\small\textbf{Output:}] \small A cover $\C$, of size $\card{P}+1$.
\end{description}
%\vspace{-.6cm}
\begin{codebox}
\Procname{$\proc{Open-Cover}(\K,P)$}
 \li	$\id{\C} \gets \emptyset$
 \li  \Parfor\,$ \sigma \gets \sigma_1$ \To $\sigma_m \in \K$
 \li 	   \Do  $\id{\C}[\sigma] \gets \proc{Partition-Cell}(P,\sigma)$
          \End
 \li \Return \id{\C}
  \End
\end{codebox}
\end{minipage}
\begin{minipage}[b]{.45\textwidth}
\begin{description}
\addtolength{\itemsep}{-.65\baselineskip}
\item[\small \textbf{Input:}] \small A graph partition $P$ of size $p$, and simplex $\sigma$
\item[\small \textbf{Output:}] \small The index $i \in [p]$ of  $\tilde{\C}$ to place $\sigma$.
\end{description}
%\vspace{-.65cm}
\begin{codebox}
\Procname{$\proc{Partition-Cell}(P, \sigma = [v_0, \ldots, v_d])$}
 \li	$R \gets \emptyset$
 \li  \For $v \gets v_0$ \To $v_d \in \sigma$
 \li 	   \Do $R \gets P(v_0)$
          \End
 \li	\If $\card{R} = 1$ \Return $R[0]$
 \li    \Else \Return $\card{P}$
\end{codebox}
\end{minipage}
\caption{The pseudocode for $\proc{Open-Cover}$ which runs in $O(md/p)$ time, where $m$ is the 
number of simplices in $\K$ a $d$-dimensional complex, and $p$ is the 
maximum number of available cores. $P$ is indexed starting at 0. For a vertex $v$, $P(v)$ 
denotes the index of the partition set of $P$ containing $v$.
}
\label{alg:open-cov}
\end{figure}
In this section we describe an algorithm for generating covers on an 
arbitrary complex from a partition of its one skeleton. We emphasize that while we propose a specific
algorithm for generating covers any procedure for generating covers suffices. In many situations there might be a 
better approach for generating covers than the one presented. Recall that in the worst case, a cover may produce an exponentially large blowup. 
However, the heuristic presented in this section guarantees that $\factor < 3$. 

There are many algorithms for generating covers, and they are all valid inputs to our parallel algorithms. 
Zomorodian \& Carlsson consider two methods for cover enumeration, \emph{random $\epsilon$-balls} 
and \emph{tilings} \cite{zc-lh-08}. For complexes embedded in a low dimensional space one might consider
algorithms based on \emph{Voronoi diagrams} or when the data is available by level sets of 
\emph{Morse functions}.  However, in the general setting it is possible to generate a cover of an 
arbitrary simplicial complex from a partition of its one skeleton with a simple intersection pattern.

The algorithm $\proc{Partition-Based-Cover}$, illustrated in Figure~(\ref{fig:vignette1}), takes a complex $\K$ and positive 
integer $p \geq 2$ as input and produces a cover $\C$ of size $p+1$ as output. 
First, we extract the one-skeleton of $\K$ and represent it as a graph $G$. 
Second, we find a graph partition $P$ of $G$ of size $p$.  
Third, we extend $P$ to an open cover $\tilde{\C}$. Finally, we extend $\tilde{\C}$ to a cover $\C$. The algorithms
for producing these two covers are called $\proc{Open-Cover}$ and
$\proc{Close-Cover}$, respectively. 

There are many algorithms for computing partitions of graphs which
seem to fall into four major classes of algorithms: geometric, non-geometric, 
spectral, and hybrid methods \cite{fj-gp-98}. Hybrid methods mix the techniques of the 
other three. In practice, we use \textsc{Metis}, a hybrid method, since it tends to produce balanced partitions quickly~\cite{KaKu95}. 
Of course any partitioning scheme will work. Next, we describe $\proc{Open-Cover}(\K,P)$, which extends a partition of $G$ to an open cover of $\K$. 

The procedure $\proc{Open-Cover}(\K,P)$ is given in Algorithm~\ref{alg:open-cov} and outputs an open cover $\tilde{\C} = \{\tilde{\C}_i\}_{i \in [p]}$ which is a partition of $\K$. Given a partition $P = \{P_i\}_{i \in [p-1]}$ of the vertex set of $G$ we expand $P$ to $\tilde{\C}$. Specifically, we first create sets 
$\tilde{\C} = \{\tilde{\C}_i\}_{i \in [p]}$ where a simplex $\sigma$ is placed into $\tilde{\C}_i$ for $i \in [p-1]$ 
if all of its vertices lie in $P_i$ and is added to $\tilde{\C}_p$ otherwise. 

In the procedure $\proc{Close-Cover}$ we replace $\tilde{\C}_{i}$ with $\C_i = \Cl{(\tilde{\C}_i)}$. However, $\tilde{\C}_i$ is closed for $i \in [p-1]$ by construction so we only close the last set.  Both $\proc{Open-Cover}$ and $\proc{Close-Cover}$ can be implemented in parallel. 
We have the following lemma: 
\begin{lemma}
\label{lem:blowup-factor}
Given a complex $\K$, $p \geq 2$, $\proc{Partition-Based-Cover}(\K, p)$ 
generates a cover $\C$ with $\factor < 3$. 
\end{lemma}
\begin{proof}
For a complex $\K$ and $p \geq 2$ let $\C$ be the cover of $\K$ by $p+1$ subcomplexes output by \proc{Partition-Based-Cover}(\K,p). The first
$p$ cover sets are disjoint since they are formed from disjoint sets of vertices. Therefore there can be at most pairwise intersections.
It follows by Lemma~\ref{lem:char-blowup-sol} that $\factor < 3$.
\end{proof}

Since we are interested only in the homology of $\K$ and not it's persistent homology we may avoid the construction of the blowup complex and use the open cover generated to place a filtration on $\K$. In particular, consider the filtration on $\K$ obtained by ordering $\tilde{\C}_i < \tilde{\C}_p$ for $i \in [p-1]$. It is clear that before including $\tilde{\C}_p$ the complex is again disconnected and thus these columns of the matrix may be reduced in parallel. Finally, we reduce this last set of columns against the columns from the first $p$ cover sets. We call this procedure $\proc{Heuristic-MH}$.
 
In the next section we compare these two parallel algorithms against the standard serial algorithm as well as the algorithm $\proc{Chunk}$ of Bauer et. al\@ on a series of examples. The $\proc{Chunk}$ algorithm is based on the spectral sequence of a filtration~\cite{bkr-cccph-13}.

\section{Experiments}
\begin{figure}[h]
\centering
\begin{subfigure}[b]{.45\textwidth}
\centering
\begin{tikzpicture}[scale=.65]
\begin{axis}[xlabel=\# of partitions, minor y tick num={1}, ylabel=speedup factor, legend style={legend pos=north west, font=\small}]
\legend{\multiblob, \bunny, \clique, \gnp, \sphere, ideal}
\addplot table [x=num_partitions, skip coords between index={0}{1}, y=speedup, col sep=comma, ignore chars=']
{pgf-speedup-figs/results/concurrent_homology/clique.11.22720.csv};
\addplot table [x=num_partitions, skip coords between index={0}{1}, y=speedup, col sep=comma, ignore chars=']
{pgf-speedup-figs/results/concurrent_homology/bunny..05.csv};
\addplot table [x=num_partitions, skip coords between index={0}{1}, y=speedup, col sep=comma, ignore chars=']
{pgf-speedup-figs/results/concurrent_homology/clique.20.csv};
\addplot table [x=num_partitions, skip coords between index={0}{1}, y=speedup, col sep=comma, ignore chars=']
{pgf-speedup-figs/results/concurrent_homology/gnp.1250.047.csv};
\addplot table [x=num_partitions, skip coords between index={0}{1}, y=speedup, col sep=comma, ignore chars=']
{pgf-speedup-figs/results/concurrent_homology/sphere.csv};
\addplot[dash pattern=on 4pt off 1pt on 4pt off 4pt, domain=2:10]{x+1};
\end{axis}
\end{tikzpicture}
\caption{Speedup factor for reducing $\partial_{\K^{\C}}$}
\label{fig:blowup-homology-speedup}
\end{subfigure}
\hfill
\begin{subfigure}[b]{.45\textwidth}
\centering
\begin{tikzpicture}[scale=.65]
%\pgfplotsset{ymax=5}
\begin{axis}[xlabel=\# of partitions/threads, minor y tick num={1}, ylabel=speedup factor, legend style={legend pos=north west, font=\small}]
\legend{\multiblob, \bunny, \clique, \gnp, \sphere, ideal}
\addplot table [x=num_partitions, skip coords between index={0}{1}, y=speedup, col sep=comma, ignore chars=']
{pgf-speedup-figs/results/cover_homology/clique.11.22720.csv};
\addplot table [x=num_partitions, skip coords between index={0}{1}, y=speedup, col sep=comma, ignore chars=']
{pgf-speedup-figs/results/cover_homology/bunny..05.csv};
\addplot table [x=num_partitions, skip coords between index={0}{1}, y=speedup, col sep=comma, ignore chars=']
{pgf-speedup-figs/results/cover_homology/clique.20.csv};
\addplot table [x=num_partitions, skip coords between index={0}{1}, y=speedup, col sep=comma, ignore chars=']
{pgf-speedup-figs/results/cover_homology/gnp.1250.047.csv};
\addplot table [x=num_partitions, skip coords between index={0}{1}, y=speedup, col sep=comma, ignore chars=']
{pgf-speedup-figs/results/cover_homology/sphere.csv};
\addplot[dash pattern=on 4pt off 1pt on 4pt off 4pt, domain=2:10]{x};
\end{axis}
\end{tikzpicture}
\caption{Speedup factor for reducing $\partial_{\K}$}
\label{fig:nonblowup-homology-speedup}
\end{subfigure}
\caption{(Left) Speedup factor $T_s/T_p$ where $T_p$ is the time to reduce $\partial_{\K^\C}$ in parallel on $p+1$ threads  and $T_s$ is the time to reduce $\partial_K$ in serial. (Right) Speedup factor of $T_s/T_p$ where $T_p$ measures the time to reduce $\partial_{\K}$ in parallel.}
\label{fig:reduction-speedup}
\end{figure}
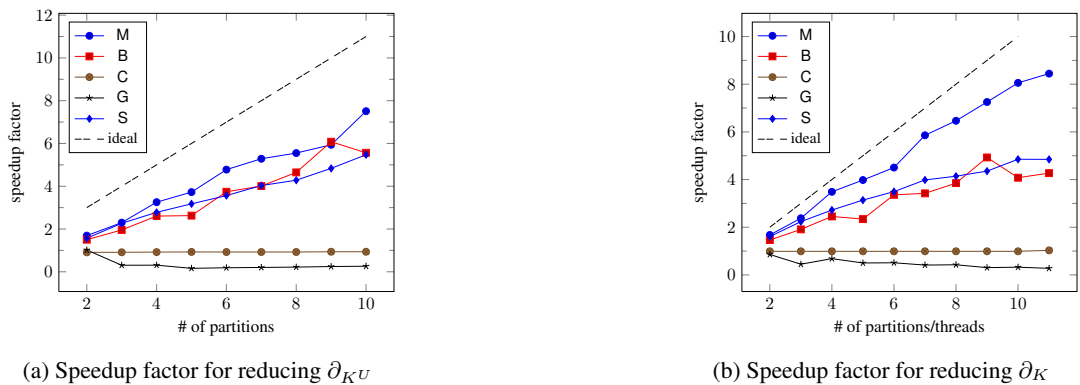
\label{sec:exp}
In this section, we describe the implementation of our algorithms
and explore their performance on real and synthetic data. We compare our
performance against our existing serial software as well as the Persistent Homology Algorithm Toolbox (PHAT)~\cite{bkr-cccph-13}. 
Our implementation is in \cplusplus\, using the generic programming paradigm. 
We rely on the METIS library for computing graph partitions 
\cite{KaKu95}, the Intel Threading Building Blocks 
Library~\cite{IntelTBB} for parallelism, and our own
library for homology computation. Our parallel implementation of $\proc{Multicore-Homology}$ computes
an initial filtration on $\K$, a cover $\C$, builds a blowup complex $\K^{\C}$ 
with its associated filtration [in parallel], and then reduces $\partial_{\K^{\C}}$.
For $\proc{Heuristic-MH}$ we reduces a permuted $\partial_\K$, instead of building $\K^{\C}$.
Unlike the psuedo-code for $\proc{Multicore-Homology}$ when reducing $\partial_{\K^{\C}}$,
our implementation reduces the columns corresponding to cells of the form $\sigma \times \tau$ with $\dim{(\tau)} > 0$ in serial
after the parallel reduction of all other cells. Preliminary experiments suggested that
this added parallelism would not produce speedup. Our serial implementation only 
computes an identical initial filtration, and then reduces $\partial_\K$.

We now provide details on how these experiments were carried out.
As previously mentioned all of our experiments are done 
using 11 cores on a 2 CPU, 12 Core, x86-64 Linux Machine, with 2.93 GHz Intel 
Xeon X5670 Processors, 74 GB of RAM, and hyperthreading disabled.
We time both parallel and serial programs in wall-clock time using the tbb::clock. 
We measure the total amount of memory requested by a process, its \emph{resident set size},
via the \emph{process filesystem}. This is an upper bound on the total memory used.  
Each time measured is the \emph{makespan} or longest running thread time within a section of code. 
Time is always reported in seconds, and all reported measurements are averaged over 10 trials.
We remind the reader that while we may spawn $p$ threads we only ever have at most $p-1$ of the total $p$ cores 
in order to leave room for system processes. In this work we use at most one thread per available 
core. When running PHAT we used the latest stable version 1.4 and the ``vector vector" option as 
this is the same basic data structure we use in our library. All software has compiled with gcc and optimizations
enabled.

\subsection{Data}
\begin{table*}
  \begin{center}
    \small
    \begin{tabular}{crrrrrrrr}
      \hline
      \multicolumn{6}{c}{Input Statistics} \\
      \multicolumn{1}{c}{$D$} &
      \multicolumn{1}{c}{$\card{D}$} &
      \multicolumn{1}{c}{$\epsilon$, $p$} &
      \multicolumn{1}{c}{$\card{E}$} &
      \multicolumn{1}{c}{$d$} &
      \multicolumn{1}{c}{$\card{\K}$} \\
      \hline
      \blobs &  249,920 & - & 1,272,319 & 10 & 46,530,559 \\
      \clique &  20 & - & 190 & 19 & 1,048,575 \\
      \bunny & 34,837 & 0.05 & 489,876 & 3 & 9,714,912 \\
      \sphere & 50,000 & 0.18 & 546,388 & 8 & 19,134,612 \\
      \gnp & 1250 & 0.047 &  & 4 & 73,309 \\
      \hline
    \end{tabular}
  \end{center}
  \caption{%
    Input Statistics: The name $D$, and number of vertices of each data set $\card{D}$, 
    as well as input parameter $\epsilon$ or $p$ in the case of a random graph, 
    embedding dimension $d = \dim{D}$, size $\card{K}$, and edge-set size $\card{E}$ of each complex $\K$. 
  }
  \label{tab:data}
\end{table*}
 
We summarize each data set in Table~\ref{tab:data}.  All complexes are skeleta of a Vietoris-Rips Complex~\cite{z-fcv-10}.  
Next, we describe the input space for each experiment. 
Recall that {\blobs} is a collection of 22,720 copies of a fully connected 10 dimensional complex on 11 vertices, organized
into 10 groups of 2,272, with each 
copy within a group connected to the next by a single edge, and each group connected to the next by a single edge as shown in Figure~(\ref{fig:blobs-vis}). 
{\clique} is a fully connected complex on 19 vertices. Recall that $\Delta^{[n]}$ has $\Theta(2^n)$ 
faces. {\bunny} is a 3-complex built on a set of points sampled from the 
\emph{Stanford bunny}. We create {\sphere} by using Muller's method~\cite{m-nmgpuns-59} to sample 
uniformly on the unit 3-sphere and then use the diagonal map
$x \rightarrow (x,x)$ to embed the points in $\R^8$~\cite{hatcher}. {\gnp} is a 
4-dimensional clique complex built on a sparse \Erdos-\Renyi\ graph $G(n,p)$ 
with $n = 1250$ and $p = 0.047$.

\subsection{Statistics}
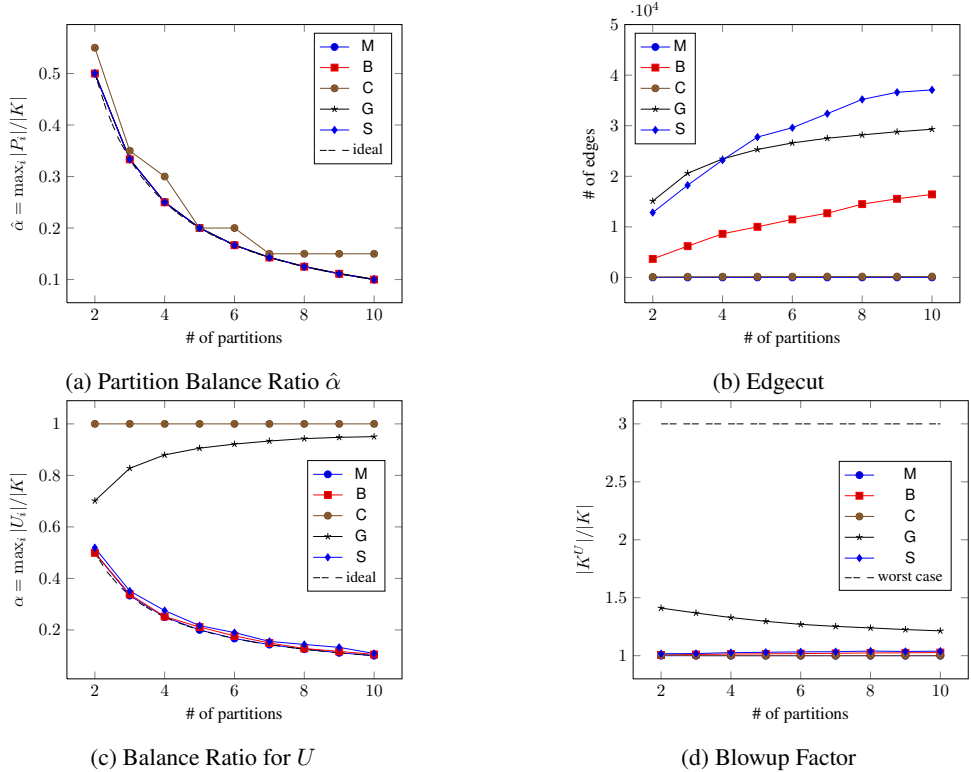
\begin{figure}
\centering
\begin{subfigure}[b]{.45\textwidth}
\centering
\begin{tikzpicture}[scale=.65]
\begin{axis}[xlabel=\# of partitions, ylabel={$\hat{\alpha}= \max_i \card{P_i} / \card{\K}$}, legend style={legend pos=north east, font=\small}]
\legend{\multiblob, \bunny, \clique, \gnp, \sphere, ideal}
\addplot table [x=num_partitions, y=graph_balance_ratio, col sep=comma] {pgf-speedup-figs/results/concurrent_homology/clique.11.22720.csv};
\addplot table [x=num_partitions, y=graph_balance_ratio, col sep=comma] {pgf-speedup-figs/results/concurrent_homology/bunny..05.csv};
\addplot table [x=num_partitions, y=graph_balance_ratio, col sep=comma] {pgf-speedup-figs/results/concurrent_homology/clique.20.csv};
\addplot table [x=num_partitions, y=graph_balance_ratio, col sep=comma] {pgf-speedup-figs/results/concurrent_homology/gnp.1250.047.csv};
\addplot table [x=num_partitions, y=graph_balance_ratio, col sep=comma] {pgf-speedup-figs/results/concurrent_homology/sphere.csv};
\addplot[dash pattern=on 4pt off 1pt on 4pt off 4pt, domain=2:10]{1/x};
\end{axis}
\end{tikzpicture}
\caption{Partition Balance Ratio $\hat{\alpha}$}
\label{fig:graph-balance}
\end{subfigure}
\begin{subfigure}[b]{.45\textwidth}
\centering
\begin{tikzpicture}[scale=.65]
%\pgfplotsset{ymax=5}
\begin{axis}[xlabel=\# of partitions, ymax=50000, ylabel={\# of edges}, legend style={legend pos=north west, font=\small}]
\legend{\multiblob, \bunny, \clique, \gnp, \sphere, ideal}
\addplot table [x=num_partitions, skip coords between index={0}{1}, y=edgecut, col sep=comma, ignore chars=']
{pgf-speedup-figs/results/concurrent_homology/clique.11.22720.csv};
\addplot table [x=num_partitions, skip coords between index={0}{1}, y=edgecut, col sep=comma, ignore chars=']
{pgf-speedup-figs/results/concurrent_homology/bunny..05.csv};
\addplot table [x=num_partitions, skip coords between index={0}{1}, y=edgecut, col sep=comma, ignore chars=']
{pgf-speedup-figs/results/concurrent_homology/clique.20.csv};
\addplot table [x=num_partitions, skip coords between index={0}{1}, y=edgecut, col sep=comma, ignore chars=']
{pgf-speedup-figs/results/concurrent_homology/gnp.1250.047.csv};
\addplot table [x=num_partitions, skip coords between index={0}{1}, y=edgecut, col sep=comma, ignore chars=']
{pgf-speedup-figs/results/concurrent_homology/sphere.csv};
\end{axis}
\end{tikzpicture}
\caption{Edgecut}
\label{fig:graph-edgecut}
\end{subfigure}
\begin{subfigure}[b]{.45\textwidth}
\centering
\begin{tikzpicture}[scale=.65]
\begin{axis}[xlabel=\# of partitions, ylabel={$\alpha = \max_i \card{\C_i} / \card{\K}$}, legend style={legend pos=north east, font=\small},legend style={at={(.95,.55)},anchor=east}]
\legend{\multiblob, \bunny, \clique, \gnp, \sphere, ideal}
\addplot table [x=num_partitions, y=cover_balance_ratio, col sep=comma,skip coords between index={0}{1}] 
{pgf-speedup-figs/results/concurrent_homology/clique.11.22720.csv};
\addplot table [x=num_partitions, y=cover_balance_ratio, col sep=comma, skip coords between index={0}{1}] 
{pgf-speedup-figs/results/concurrent_homology/bunny..05.csv};
\addplot table [x=num_partitions,skip coords between index={0}{1}, y=cover_balance_ratio, col sep=comma] 
{pgf-speedup-figs/results/concurrent_homology/clique.20.csv};
\addplot table [x=num_partitions, skip coords between index={0}{1},y=cover_balance_ratio, col sep=comma] 
{pgf-speedup-figs/results/concurrent_homology/gnp.1250.047.csv};
\addplot table [x=num_partitions, skip coords between index={0}{1},y=cover_balance_ratio, col sep=comma] 
{pgf-speedup-figs/results/concurrent_homology/sphere.csv};
\addplot[dash pattern=on 4pt off 1pt on 4pt off 4pt, domain=2:10]{1/x};
\end{axis}
\end{tikzpicture}
\caption{Balance Ratio for $\C$}
\label{fig:balance-factors}
\end{subfigure}
\begin{subfigure}[b]{.45\textwidth}
\centering
\begin{tikzpicture}[scale=.65]
\begin{axis}[xlabel=\# of partitions, ylabel=$\ratio$, legend style={legend pos=north east, font=\small},legend style={at={(.95,.55)},anchor=east}]
\legend{\multiblob, \bunny, \clique, \gnp, \sphere, worst case}
\addplot table [x=num_partitions, y=blowup_factor, col sep=comma,skip coords between index={0}{1}] 
{pgf-speedup-figs/results/concurrent_homology/clique.11.22720.csv};
\addplot table [x=num_partitions, y=blowup_factor, col sep=comma, skip coords between index={0}{1}] 
{pgf-speedup-figs/results/concurrent_homology/bunny..05.csv};
\addplot table [x=num_partitions,skip coords between index={0}{1}, y=blowup_factor, col sep=comma] 
{pgf-speedup-figs/results/concurrent_homology/clique.20.csv};
\addplot table [x=num_partitions, skip coords between index={0}{1},y=blowup_factor, col sep=comma] 
{pgf-speedup-figs/results/concurrent_homology/gnp.1250.047.csv};
\addplot table [x=num_partitions, skip coords between index={0}{1},y=blowup_factor, col sep=comma] 
{pgf-speedup-figs/results/concurrent_homology/sphere.csv};
\addplot[dash pattern=on 4pt off 1pt on 4pt off 4pt, domain=2:10]{3};
\end{axis}
\end{tikzpicture}
\caption{Blowup Factor}
\label{fig:blowup-factors}
\end{subfigure}
\caption{Statistics for partitions and covers generated}
\label{fig:statistics}
\end{figure}
Recall from Section~\ref{sec:partition-based-covers} that our input is a complex
$\K$ and integer $p > 1$. Our goal is to build a balanced cover for 
which $\factor$ is as small as possible. First, we build a a graph partition of
the one skeleton $G(\K)$. To produce our graph partition we chose the
unsupervised graph partitioning algorithm METIS because it tends to produce 
balanced graph partitions. In Figure~(\ref{fig:graph-balance}) we show the balance
ratio $\hat{\alpha} = \max_i{\card{V_i}}/\card{V}$ for each partition produced by METIS.
Next, we complete our graph partition into a cover. Figure~(\ref{fig:balance-factors})
shows the balance ratio $\alpha = \max_i{\card{\C_i}}/\card{\K}$ for covers
produced by: $\proc{Partition-Based-Cover}$.
Finally, the procedure $\proc{Build-Blowup-Complex}$ 
computes the blowup complex along with its filtration. 
In Figure~(\ref{fig:blowup-factors}) we plot $\factor$. Recall that 
covers produced by $\proc{Partition-Based-Cover}$ have
$\factor < 3$ and in general for $n$ sets this ratio is at worst $O(2^n)$. 

\subsection{Timing \& Measurements}
For each of our data sets we present the speedup factor of
our reduction algorithm versus serial persistence in Figure~(\ref{fig:reduction-speedup}).

First, we can see that our techniques tend to scale the best on inputs
in which all topological features are localized by the cover. For example, we see the
best performance on $\blobs$. This is not surprising since for any $p \in [2, 10]$ 
this complex exhibits a partition-based cover which balances its 
46.5M simplices nearly perfectly while maintaining that the size of all intersections between all sets is exactly $p-1$.
Second, geometric inputs such as $\bunny$ and $\sphere$ have entirely global topology; These
global topological features are resolved by reducing a handful of columns in the portion of the computation that is 
executed serially. However, these inputs still emit balanced covers, so we see speedup since 
overall the bulk of the work is roughly evenly divided across each core.
Finally, we see that inputs which are flag complexes of cliques or expander graphs, 
such as $\clique$ or $\gnp$, emit no balanced cover and all covers seem to result in a large blowup complex. 
As expected our parallel algorithms exhibit no speedup on these inputs. 
\begin{figure}
\centering
%\fbox{
\begin{tikzpicture}[scale=.65]
\begin{semilogyaxis}[
name=plot1,
%legend pos=outer north east,
xmin=0,
xmax=11,
ymin=1,
ymax=20000,
xlabel=\# of partition sets,
ylabel=Maximum Resident Set Size (MB),
title={$\proc{Multicore-Homology}$},
legend style={at = {(.9,.5)},font=\large}]
%\legend{\multiblob, \bunny, \clique, \gnp, \sphere}
\addplot table [x=num_partitions, y=max memory (MB), col sep=comma] {pgf-speedup-figs/results/concurrent_homology/clique.11.22720.csv};
\addplot table [x=num_partitions, y=max memory (MB), col sep=comma] {pgf-speedup-figs/results/concurrent_homology/bunny..05.csv};
\addplot table [x=num_partitions, y=max memory (MB), col sep=comma] {pgf-speedup-figs/results/concurrent_homology/clique.20.csv};
\addplot table [x=num_partitions, y=max memory (MB), col sep=comma] {pgf-speedup-figs/results/concurrent_homology/gnp.1250.047.csv};
\addplot table [x=num_partitions,, y=max memory (MB), col sep=comma] {pgf-speedup-figs/results/concurrent_homology/sphere.csv};
\end{semilogyaxis}
\begin{semilogyaxis}[
name=plot2,
at=(plot1.outer east), anchor=outer west,
xmin=0,
xmax=11,
ymin=1,
ymax=20000,
xlabel=\# of threads,
title={$\proc{Chunk}$},
legend style={legend pos=south east,font=\tiny}, 
legend style={at = {(1,.45)},font=\large}]
%\legend{\multiblob, \bunny, \clique, \gnp, \sphere}
\addplot table [x=num_threads, y=max memory (MB), col sep=comma] {pgf-speedup-figs/results/phat_14_chunk/clique.11.22720.csv};
\addplot table [x=num_threads, y=max memory (MB), col sep=comma] {pgf-speedup-figs/results/phat_14_chunk/bunny..05.csv};
\addplot table [x=num_threads, y=max memory (MB), col sep=comma] {pgf-speedup-figs/results/phat_14_chunk/clique.20.csv};
\addplot table [x=num_threads, y=max memory (MB), col sep=comma] {pgf-speedup-figs/results/phat_14_chunk/gnp.1250.047.csv};
\addplot table [x=num_threads, y=max memory (MB), col sep=comma] {pgf-speedup-figs/results/phat_14_chunk/sphere.csv};
\end{semilogyaxis}

\begin{semilogyaxis}[
name=plot3,
at=(plot1.below south west), anchor=above north west,
xmax=11,
ymin=1,
ymax=20000,
title={$\proc{Heuristic-MH}$},
xlabel=\# of partition sets, 
legend style={at={(.98,.43)}, font=\footnotesize}]
\legend{\multiblob, \bunny, \clique, \gnp, \sphere}
\addplot table [x=num_partitions, y=max memory (MB), col sep=comma] {pgf-speedup-figs/results/cover_homology/clique.11.22720.csv};
\addplot table [x=num_partitions, y=max memory (MB), col sep=comma] {pgf-speedup-figs/results/cover_homology/bunny..05.csv};
\addplot table [x=num_partitions, y=max memory (MB), col sep=comma] {pgf-speedup-figs/results/cover_homology/clique.20.csv};
\addplot table [x=num_partitions, y=max memory (MB), col sep=comma] {pgf-speedup-figs/results/cover_homology/gnp.1250.047.csv};
\addplot table [x=num_partitions, y=max memory (MB), col sep=comma] {pgf-speedup-figs/results/cover_homology/sphere.csv};
\end{semilogyaxis}

\begin{semilogyaxis}[
name=plot4,
at=(plot3.outer east), anchor=outer west,
ymin=1,
ymax=20000,
xmin=0,
xmax=11,
xlabel= \# of threads, 
title={$\proc{Spectral-Sequence}$},
minor y tick num={5},
legend style={at={(1,.6)}, font=\large}]
\addplot table [x=num_threads, y=max memory (MB), col sep=comma] {pgf-speedup-figs/results/phat_14_ss/clique.11.22720.csv};
\addplot table [x=num_threads, y=max memory (MB), col sep=comma] {pgf-speedup-figs/results/phat_14_ss/bunny..05.csv};
\addplot table [x=num_threads, y=max memory (MB), col sep=comma] {pgf-speedup-figs/results/phat_14_ss/clique.20.csv};
\addplot table [x=num_threads, y=max memory (MB), col sep=comma] {pgf-speedup-figs/results/phat_14_ss/gnp.1250.047.csv};
\addplot table [x=num_threads, y=max memory (MB), col sep=comma] {pgf-speedup-figs/results/phat_14_ss/sphere.csv};
\end{semilogyaxis}
\end{tikzpicture}
\caption{Total memory usage for each algorithm. Recall that PHAT takes as input a boundary matrix whereas the 
procedures outlined in this work take as input a simplicial complex and generates an identical boundary matrix before reducing it. 
At $x=1$ on all plots we display the memory used for the standard algorithm from the appropriate software package.}
\label{fig:memory-usage}
\end{figure}
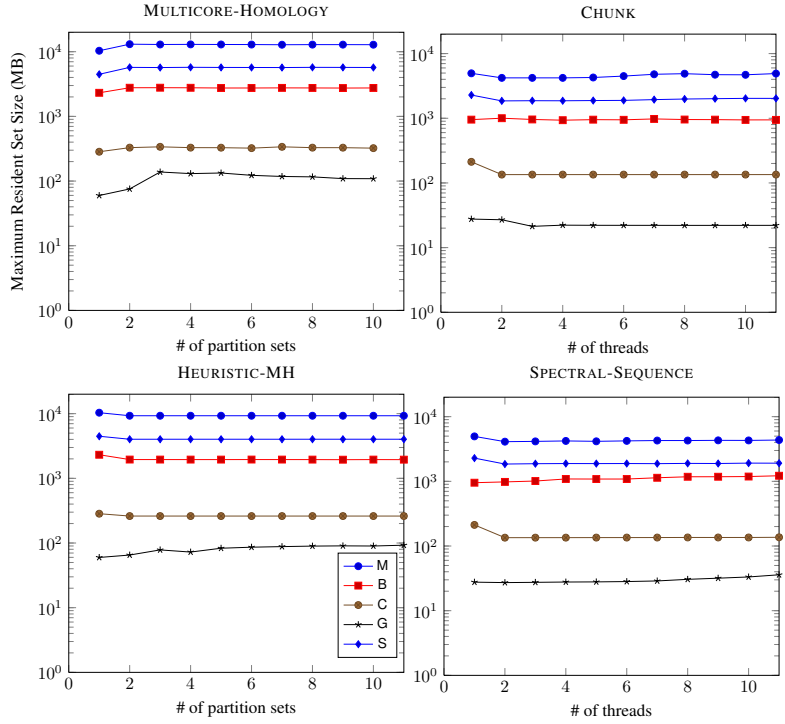

We observe that with the exception of $\gnp$ the parallel reduction of the boundary matrix for the blowup complex 
runs in time similar to the parallel reduction of the permuted boundary matrix. However there is overhead to each approach. 
Both algorithms require the computation of a cover. On one hand, to reduce $\partial_{\K^\C}$ we must first build $\K^\C$ and its associated filtration. 
However in $\proc{Heuristic-MH}$ we must construct a new filtration on $\K$. 
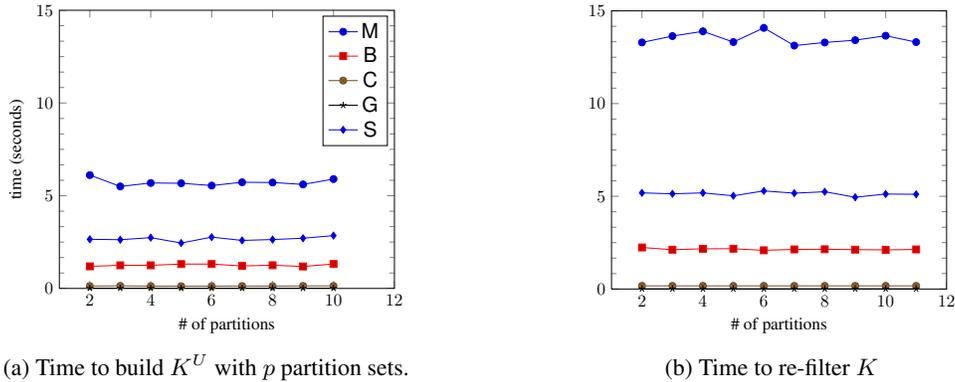
\begin{figure}
\centering
\begin{subfigure}[b]{.45\textwidth}
\centering
\begin{tikzpicture}[scale=.65]
\begin{axis}[
name=left axis,
ymin=0,
ymax=15,
xmax=12,
xlabel=\# of  partitions,
ylabel=time (seconds),
minor y tick num={5},
legend style={font=\large}]
\legend{\multiblob, \bunny, \clique, \gnp, \sphere, ideal}
\addplot table [x=num_partitions, skip coords between index={0}{1}, y=build_blowup, col sep=comma, ignore chars=']
{pgf-speedup-figs/results/concurrent_homology/clique.11.22720.csv};
\addplot table [x=num_partitions, skip coords between index={0}{1}, y=build_blowup, col sep=comma, ignore chars=']
{pgf-speedup-figs/results/concurrent_homology/bunny..05.csv};
\addplot table [x=num_partitions, skip coords between index={0}{1}, y=build_blowup, col sep=comma, ignore chars=']
{pgf-speedup-figs/results/concurrent_homology/clique.20.csv};
\addplot table [x=num_partitions, skip coords between index={0}{1}, y=build_blowup, col sep=comma, ignore chars=']
{pgf-speedup-figs/results/concurrent_homology/gnp.1250.047.csv};
\addplot table [x=num_partitions, skip coords between index={0}{1}, y=build_blowup, col sep=comma, ignore chars=']
{pgf-speedup-figs/results/concurrent_homology/sphere.csv};
\end{axis}
\end{tikzpicture}
\caption{Time to build $\K^{\C}$ with $p$ partition sets.}
\end{subfigure}
\begin{subfigure}[b]{.45\textwidth}
\centering
\begin{tikzpicture}[scale=.65]
\begin{axis}[
ymin=0,
ymax=15,
xmax=12,
xlabel=\# of  partitions,
%ylabel=time (seconds),
minor y tick num={5},
legend style={legend pos = north west,font=\large}]
%\legend{\multiblob, \bunny, \clique, \gnp, \sphere}
\addplot table [x=num_partitions, skip coords between index={0}{1}, y=re-filter complex, col sep=comma, ignore chars=']
{pgf-speedup-figs/results/cover_homology/clique.11.22720.csv};
\addplot table [x=num_partitions, skip coords between index={0}{1}, y=re-filter complex, col sep=comma, ignore chars=']
{pgf-speedup-figs/results/cover_homology/bunny..05.csv};
\addplot table [x=num_partitions, skip coords between index={0}{1}, y=re-filter complex, col sep=comma, ignore chars=']
{pgf-speedup-figs/results/cover_homology/clique.20.csv};
\addplot table [x=num_partitions, skip coords between index={0}{1}, y=re-filter complex, col sep=comma, ignore chars=']
{pgf-speedup-figs/results/cover_homology/gnp.1250.047.csv};
\addplot table [x=num_partitions, skip coords between index={0}{1}, y=re-filter complex, col sep=comma, ignore chars=']
{pgf-speedup-figs/results/cover_homology/sphere.csv};
\end{axis}
\end{tikzpicture}
\caption{Time to re-filter $\K$}
\end{subfigure}
\caption{Comparison of the time to build a blowup complex in $O(\frac{m}{p} + p)$ time versus re-filter the base complex in $O(\frac{m}{p}\log{m})$.}
\label{fig:blowup-vs-no-blowup}
\end{figure}
Recall that the procedure $\proc{Build-Blowup-Complex}$ runs in parallel and has parallel running time $O(2m/p + p)$ time where $m = \card{\K^{\C}}$ 
and $p$ is the number of processors available. 
The procedure $\proc{Build-Blowup-Complex}$ is implemented as a variant of the $\proc{Prefix-Sum}$ algorithm~\cite{breshears}. 
In particular this means that $\proc{Build-Blowup-Complex}$ produces the filtration of the blowup complex along with the complex itself. 
Aside from its output $\proc{Build-Blowup-Complex}$ only uses $O(p)$ extra space. When avoiding the blowup complex we 
do so by creating a new filtration in $O(\frac{m}{p}\log{m})$ where $m = \card{\K}$ and $p$ is the total number of available threads.

Figure~(\ref{fig:blowup-vs-no-blowup}) compares the running time of $\proc{Build-Blowup-Complex}$ against the time to re-filter $\K$.
From the standpoint of memory consumption it is clear that the blowup avoiding algorithm is a better choice. However,
when the resulting blowup complex is similar in size to the original space, It may be possible to significantly improve overall running time 
by building the blowup complex simply because the process of sorting may end up being slower than building the blowup.

We end this section by comparing the \mv algorithm to $\proc{Chunk}$ and \proc{Spectral-Sequence} algorithms available in PHAT.
$\proc{Spectral-Sequence}$ and $\proc{Chunk}$ are parallel implementations of the spectral sequence algorithm based on the spectral sequence of a filtration~\cite{bkr-cccph-13}.
We plot the time to reduce $\partial_{\K}$ and $\partial_{\K^\C}$ with $p$ threads versus the time for the each algorithm from PHAT to reduce 
$\partial_{\K}$ in Figure~(\ref{fig:ctl_vs_phat}). 
Figure~(\ref{fig:memory-usage}) compares the total memory usage for these algorithms. Recall that PHAT takes as input
a description of $\partial_{\K}$ whereas for our experiments we read in as input $\K$ and then build and reduce $\partial_{\K^{\_}}$. 
While the implementation of the chunk algorithm in PHAT can be significantly faster than its implementation of the standard algorithm,
their algorithms do not always seem to scale with the number of available threads. Our experiments suggest that the algorithms 
provided in PHAT attain speedup mainly due to the out of order nature of their reductions. The two optimizations used
in these algorithms significantly reduces the total work required as compared to the serial algorithm, but these optimizations 
do not seem to help scalability.  Practically, this software is still in the early stages of development, so we expect future versions to 
be more competitive.
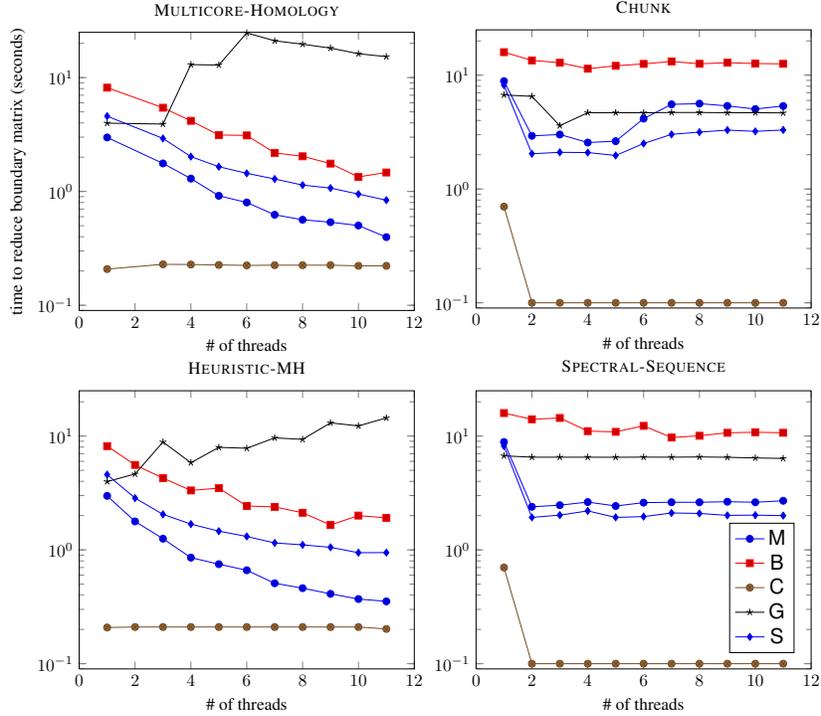
\begin{figure}
\centering
\begin{tikzpicture}[scale=.65]
\begin{semilogyaxis}[
name=plot1,
ymin=.09,
ymax=25,
xmin=0,
xmax=12,
xlabel=\# of  threads, 
title={$\proc{Multicore-Homology}$},
ylabel= time to reduce boundary matrix (seconds),
minor y tick num={5},
legend style={at = {(.25,.6)}, font=\large}]
%\legend{\multiblob, \bunny, \clique, \gnp, \sphere}
\addplot table [x=num_threads, y=persistence, col sep=comma] {pgf-speedup-figs/results/concurrent_homology/clique.11.22720.csv};
\addplot table [x=num_threads, y=persistence, col sep=comma] {pgf-speedup-figs/results/concurrent_homology/bunny..05.csv};
\addplot table [x=num_threads, y=persistence, col sep=comma] {pgf-speedup-figs/results/concurrent_homology/clique.20.csv};
\addplot table [x=num_threads, y=persistence, col sep=comma] {pgf-speedup-figs/results/concurrent_homology/gnp.1250.047.csv};
\addplot table [x=num_threads, y=persistence, col sep=comma] {pgf-speedup-figs/results/concurrent_homology/sphere.csv};
\end{semilogyaxis}
\begin{semilogyaxis}[
name=plot3,
at=(plot1.below south east), anchor=above north east,
ymin=.09,
ymax=25,
xmin=0,
xmax=12,
xlabel= \# of threads,
title={$\proc{Heuristic-MH}$},
%ylabel= time to reduce boundary matrix (seconds),
minor y tick num={5},
legend style={at={(1,.55)}, font=\large}]
%\legend{\multiblob, \bunny, \clique, \gnp, \sphere}
\addplot table [x=num_partitions, y=persistence, col sep=comma] {pgf-speedup-figs/results/cover_homology/clique.11.22720.csv};
\addplot table [x=num_partitions, y=persistence, col sep=comma] {pgf-speedup-figs/results/cover_homology/bunny..05.csv};
\addplot table [x=num_partitions, y=persistence, col sep=comma] {pgf-speedup-figs/results/cover_homology/clique.20.csv};
\addplot table [x=num_partitions, y=persistence, col sep=comma] {pgf-speedup-figs/results/cover_homology/gnp.1250.047.csv};
\addplot table [x=num_partitions, y=persistence, col sep=comma] {pgf-speedup-figs/results/cover_homology/sphere.csv};
\end{semilogyaxis}

\begin{semilogyaxis}[
name=plot4,
at=(plot3.right of north east), anchor=left of north west,
ymin=.09,
ymax=25,
xmin=0,
xmax=12,
xlabel= \# of threads, 
title={$\proc{Spectral-Sequence}$},
minor y tick num={5},
legend style={at={(.95,.525)}, font=\large}]
\legend{\multiblob, \bunny, \clique, \gnp, \sphere}
\addplot table [x=num_partitions, y=persistence, col sep=comma] {pgf-speedup-figs/results/phat_14_ss/clique.11.22720.csv};
\addplot table [x=num_partitions, y=persistence, col sep=comma] {pgf-speedup-figs/results/phat_14_ss/bunny..05.csv};
\addplot table [x=num_partitions, y=persistence, col sep=comma] {pgf-speedup-figs/results/phat_14_ss/clique.20.csv};
\addplot table [x=num_partitions, y=persistence, col sep=comma] {pgf-speedup-figs/results/phat_14_ss/gnp.1250.047.csv};
\addplot table [x=num_partitions,, y=persistence, col sep=comma] {pgf-speedup-figs/results/phat_14_ss/sphere.csv};
\end{semilogyaxis}

\begin{semilogyaxis}[
name=plot2,
at=(plot4.above north west),
anchor = below south west,
ymin=.09,
ymax=25,
xmin=0,
xmax=12,
xlabel=\# of  threads,
title={$\proc{Chunk}$},
minor y tick num={5},
legend style={at={(.95,.525)}, font=\large}]
%\legend{\multiblob, \bunny, \clique, \gnp, \sphere}
\addplot table [x=num_partitions, y=persistence, col sep=comma] {pgf-speedup-figs/results/phat_14_chunk/clique.11.22720.csv};
\addplot table [x=num_partitions, y=persistence, col sep=comma] {pgf-speedup-figs/results/phat_14_chunk/bunny..05.csv};
\addplot table [x=num_partitions, y=persistence, col sep=comma] {pgf-speedup-figs/results/phat_14_chunk/clique.20.csv};
\addplot table [x=num_partitions, y=persistence, col sep=comma] {pgf-speedup-figs/results/phat_14_chunk/gnp.1250.047.csv};
\addplot table [x=num_partitions,, y=persistence, col sep=comma] {pgf-speedup-figs/results/phat_14_chunk/sphere.csv};
\end{semilogyaxis}
\end{tikzpicture}
\caption{Time to reduce the boundary matrix for each algorithm. At $x=1$ on all plots we display the running time for reducing $\partial_{\K}$ using the standard algorithm from the appropriate software package.}
\label{fig:ctl_vs_phat}
\end{figure}
\section{Conclusion \& Future Work}
In this paper we presented two methods for computing homology in parallel. 
We describe each step of both methods, implement all algorithms, and present 
preliminary experimental results. While our main goal is to compute the 
persistent homology of larger complexes in 
distributed memory we have demonstrated the ability for parallel computations 
based on spatial decompositions of the input to outperform serial computations.

There are many avenues for future research. 
The nerve of the covers generated in this paper have are a star graph. 
It would be useful to be able to generate covers whose nerve has higher topological features. 
For example, if the nerve was a cycle then we could take advantage of 
added parallelism when reducing the corresponding cells in the blowup complex. 
The partition based covers are akin to a bottom up approach to cover generation. A top
down algorithm which operates by partitioning the maximal cells might have better performance
on datasets where a small separator is non existent or difficult to find. It would be of clear 
interest to have an approximation algorithm to the problem discussed in this work or to
a variant thereof. It would also be of interest to combine the algorithms outlined in this works with the ones 
from PHAT. In particular, each piece of the boundary matrix produced by a Mayer-Vietoris style 
algorithm could be further reduced via these alternative approaches. 

It is possible to filter the blowup complex to have identical persistent homology 
to that of a filtration $\Filt{\K}$ of an input complex $\K$. 
Given a filtration $\Filt{\K}$ on $\K$ and a cover $\C$
one can construct a filtration on $\K^\C$ by restriction of the cover to the 
each subspace in the filtration. One can now use this data to construct a filtration of 
blowup complexes. The resulting filtration produces identical persistent homology to that 
of $\Filt{\K}$ on $\K$. At a chain level, this amounts to ordering product cells first by their factor in $\K$, breaking ties
using the second factor. Recall that in this work we ordered product cells first by the second factor, 
breaking ties using the first factor. While, it is no longer straightforward to carry out the 
persistence algorithm in parallel as described in this work, it is possible to compute the persistent homology
of this filtration in parallel. We leave the details to a followup paper.

\section*{Acknowledgments \& Bibliography}
The authors would like to thank Gunnar Carlsson, Steve Canon, and Milka Doktorova, for discussions and support.

\subsection*{Bibliography}
{ 
  \small 
  \bibliographystyle{elsarticle-num}
  \bibliography{multicore} 
}
\end{document}